\documentclass[11pt,reqno]{amsart}
\usepackage{mathrsfs,amsmath,amsxtra,amssymb,amsthm,amsfonts,subcaption}
\usepackage{mathtools}

\newcommand{\Z}{\mathcal{Z}}

\newcommand{\kla}{\left ( }
\newcommand{\mer}{\right ) }
\renewcommand{\for}{\begin{eqnarray*}}
\newcommand{\mel}{\end{eqnarray*}}
\def\fr{\begin{align*}}

\newcommand{\kl}{\pl \le \pl}
\newcommand{\gl}{\pl \ge \pl}

\newcommand{\lel}{\pl = \pl}

\newcommand{\ten}{\otimes}

\newcommand{\pl}{\hspace{.1cm}}

\newcommand{\ran}{\rangle}
\newcommand{\lan}{\langle}

\newcommand{\om}{\omega}

\newcommand{\al}{\alpha}
\newcommand{\si}{\sigma}

\newcommand{\E}{{\mathcal E}}
\newcommand{\V}{{\mathcal V}}
\newcommand{\A}{{\mathcal A}}
\newcommand{\B}{{\mathcal B}}
\newcommand{\D}{{\mathcal D}}
\newcommand{\M}{{\mathcal M}}

\newcommand{\R}{{\mathcal R}}
\renewcommand{\S}{{\mathcal S}}
\newcommand{\T}{{\mathcal T}}
\newcommand{\X}{{\mathcal X}}
\newcommand{\Y}{{\mathcal Y}}
\newcommand{\W}{{\mathcal W}}

\newcommand{\mz}{{\mathbb M}}

\newcommand{\N}{{\mathcal N}}

\newcommand{\I}{{\mathcal I}}
\newcommand{\U}{{\mathcal U}}

\newcommand{\C}{\mathcal{C}}
\newtheorem{lemma}{Lemma}[section]
\newtheorem{prop}[lemma]{Proposition}
\newtheorem{theorem}[lemma]{Theorem}

\newtheorem{cor}[lemma]{Corollary}

\newtheorem{rem}[lemma]{Remark}
\newtheorem{definition}[lemma]{Definition}

\newcommand{\re}{\begin{rem}\rm}
\newcommand{\mar}{\end{rem}}
\newtheorem{exam}[lemma]{Example}

\newcommand{\ket}[1]{|{#1}\rangle}
\newcommand{\ketbra}[1]{|{#1}\rangle\langle{#1}|}
\newcommand{\qd}{\end{proof}\vspace{0.5ex}}
\newcommand{\prf}{\begin{proof}[\bf Proof:]}

\newcommand{\xspace}{\hbox{\kern-2.5pt}}

\newtheorem*{theorem*}{Theorem}

\allowdisplaybreaks

\usepackage{verbatim}
\usepackage{amsmath}
\usepackage{amsthm}
\usepackage{listings}
\usepackage{graphicx}
\usepackage{float}
\usepackage{amsfonts}
\usepackage{latexsym}
\usepackage{color}
\usepackage{url}
\usepackage{hyperref}
\usepackage{setspace}
\usepackage[normalem]{ulem}
\usepackage{braket}
\usepackage[toc,page]{appendix}
\usepackage{cancel}
\usepackage{enumitem}

\title{Unifying Entanglement with Uncertainty via Symmetries of Observable Algebras}
\author{Li Gao, Marius Junge, Nicholas LaRacuente}

\DeclareMathOperator{\tr}{tr}

	\newcommand{\RR}{\mathbb{R}}
	\newcommand{\CC}{\mathbb{C}}
	
	\newcommand{\BB}{\mathbb{B}}

	\newcommand{\id}{\hat{1}}
	
	\newcommand{\s}{\mathfrak{S}}
	\renewcommand{\t}{\mathfrak{T}}
	
	\newcommand{\Hilbert}{\mathcal{H}}

\begin{document}
\maketitle

\begin{abstract}
Strong subadditivity goes beyond the tensored subsystem and commuting operator models. As previously noted by Petz and later by Araki and Moriya, two subalgebras of observables satisfy a generalized SSA-like inequality if they form a commuting square. We explore the interpretation and consequences in finite dimensions, connecting various entropic uncertainty relations for mutually unbiased bases with the positivity of a generalized conditional mutual information (CMI), and with inequalities on relative entropies of coherence and asymmetry. We obtain a bipartite resource theory of operations under which the two subalgebras are respectively invariant and covariant, with CMI as a monotone, and generalized non-classical monotones based on squashed entanglement and entanglement of formation. Free transformations support conversion between entanglement and uncertainty-based configurations, as ``EPR $\leftrightarrow$ 2UCR." Our theory quantifies the common non-classicality in entanglement and uncertainty, implying a strong conceptual link between these fundamentally quantum phenomena.
\end{abstract}

\section{Introduction / Background}
\noindent The strong subadditivity (SSA) of quantum entropy, proved by Lieb and Ruskai \cite{lieb_proof_1973} in 1973, is one of most fundamental inequalities in quantum information. It states that for any tripartite density $\rho^{ABC}$, the conditional mutual information (CMI) defined as follows is always non-negative:
\begin{align}
\label{eq:origssa}
I(A:B|C)_\rho \equiv H(AC)_\rho +H(BC)_\rho - H(C)_\rho - H(ABC)_\rho \geq 0 \pl,
\end{align}
where $H$ is the entropy of the reduced density. Some other important theorems, such as the data processing inequality of quantum relative entropy, are known to be equivalent to SSA. Quantum entanglement is another of the most historically important notions in the theory. The squashed entanglement \cite{christandl_squashed_2004} uses CMI to quantify entanglement,
\begin{equation}
E_{sq}(A:B)_\rho \equiv \inf_{\tilde{\rho}^{ABC} : \tilde{\rho}^{AB} = \rho^{AB}} I(A:B|C)_{\tilde{\rho}}.
\label{eq:esquash}
\end{equation}
Of similarly historic importance is the uncertainty principle for complementary observables. Quantum information can express the uncertainty principle in terms of entropy \cite{wehner_entropic_2010, coles_entropic_2017}, notably in the relatively recent uncertainty relations with quantum memory \cite{berta_uncertainty_2010},
\begin{align}
\label{eq:ucrmemory}
H(\X|B)_\rho + H(\Z|B)_\rho \ge \log \frac{1}{c} + H(A|B)_\rho,
\end{align}
in which $\X$ and $\Z$ are two complementary measurement bases of the same quantum system $A$, allowing quantum correlations between $A$ and a memory system $B$, and $c$ is the maximum overlap between bases $\X$ and $\Z$. The uncertainty relations with memory suggest links between entanglement and uncertainty. Petz and Araki-Moriya's algebraic strong subadditivity \cite{petz_certain_1991, araki_equilibrium_2003}, which we recall in Corollary \ref{cor:ssa}, suggest further connections that we discuss in section \ref{sec:connection}.

The nature of entanglement remains a matter of debate \cite{earman_puzzles_2015} even almost a century after Einstein, Podolsky \& Rosen's objection to its conflict with the completeness and locality in physical theories \cite{einstein_can_1935}. The usual tensor system $A \otimes B$ may have spatial separation, though it is neither necessary nor implied for this. Several experiments \cite{boschi_experimental_1998, michler_experiments_2000} have observed entanglement-like features of teleportation and Bell violation between different aspects of the same particle, a tensor product of co-located systems. Conversely, field theories \cite{witten_aps_2018} suggest spacelike separation without tensor products. Meanwhile, commuting operator models \cite{cleve_perfect_2017} may support an operational analog of local operations \cite{crann_state_2019} for systems not in tensor product. Subalgebraic generalizations of entanglement exist in the identical particle setting \cite{balachandran_entanglement_2013, balachandran_algebraic_2013} and even in subspaces without algebraic closure \cite{barnum_subsystem-independent_2004}. We neither claim to resolve the debate nor suggest a new definition. Rather, we quantify a form of nonclassicality that appears to extend beyond entanglement, linking it with uncertainty relations.

In particular, we study bipartite, correlation-like phenomena between algebras that need not commute with each other, requiring a commuting square condition as a looser form of independence. A party holding one algebra may perform operations that are undetectable from the other - symmetries of the other algebra. We generalize conditional mutual information (CMI) and squashed entanglement to this setting. We construct operations under which generalized CMI and its squashed counterpart are positive and non-increasing. We show that these allow conversion between two qubits of maximum uncertainty for each party and one entangled qubit pair. Positivity of generalized CMI implies strong subadditivity and several uncertainty relations for mutually unbiased bases, such as an uncertainty relation with memory (example \ref{exam:ucrmem}) and generalized Maassen-Uffink relation (Theorem \ref{B}).

A pure entangled state appears mixed in any complete, local measurement basis but reveals its purity under joint observables. For observable algebras, the idea extends to incompatible measurement bases on a single system. A pure state that appears mixed for each observable in a given set may be an eigenstate of one generated by them. The goal of this paper is to turn this heuristic into a quantitative correspondence that generalizes to mixed densities.

\section{Results}
\subsection{Notation}
We denote by $\BB(\Hilbert)$ the bounded operators on a Hilbert space $\Hilbert$. We use the capital letters $A,B,C,\cdots$ to index quantum systems $\Hilbert_A,\Hilbert_B,\Hilbert_C$ and denote $|A|=\dim \Hilbert_A$. In this paper, we consider only finite dimensional Hilbert spaces and algebras. A \emph{von Neumann algebra} $\M$ is a $*$-subalgebra of $\BB(\Hilbert)$ and isomorphic to an orthogonal sum of matrix blocks. Let $tr$ be the matrix trace. An operator $\rho\in B(\Hilbert)$ is a density if $\tau(\rho)=1,\rho\ge 0$. For the subalgebra $\M$, we denote by $S_1(\M)$ the densities in $\M$. For the subsystem $A$, $S_1(A)$ denotes the densities on $A$. Given a subalgebra $\N\subset \M$, the conditional expectation $\E_\N:\M \to \N$ is the unique completely positive trace preserving map such that 
\[\tr(ab)=\tr(\E_\N(a)b) \pl, \pl \forall \pl a\in \M, b\in \N\pl.\]
The conditional expectation $\E_\N$ sends every state $\rho$ to its restriction on $\N$. The von Neumann entropy of a density $\rho$ is defined as $H(\rho)=-\tr(\rho \log \rho )$. We denote the subalgebra entropy as $H(\N)_\rho = H(\E_\N(\rho))$. This is consistent with the subsystem notation $H(A)_\rho=H(\rho^A)$. By $D (\rho \| \si) = tr ( \rho \log \rho - \rho \log \si)$ we denote the relative entropy. For a pair of subalgebras $\S, \T \subset \M$, we denote by $\S \T = \braket{\S \cup \T}$ the algebra generated by their union. We use $\CC 1$ to denote the scalar multiple of of identity, which is the observables corresponding to phase. More information about conditional expectations for finite-dimensional von Neumann algebras is in Appendix \ref{sec:alg}.

For a tripartite state $\rho^{ABC}$, we denote by $\A=\BB(\Hilbert_A)$ the subalgebra of matrices on subsystem $A$ and similarly $\A \C=\BB(\Hilbert_A\ten \Hilbert_C)$. Within a qubit system $\BB(\Hilbert)\cong M_2$, we denote by $X$,$Y$ and $Z$ the Pauli matrices, and respectively by $\X$, $\Y$ and $\Z$ as subalgebras generated. For higher dimensional system, we may also use these letters to denote complementary, commutative subalgebras. When needed, we use a subscript to denote the restriction to a subsystem, such as $\X_A\subset \A$ for the $\X$-measurement on $\Hilbert_A$. For an algebra $\M$, the center of $\M$ is the part which commutes with all elements of $\M$, such as a classical subsystem. If $\M$'s center is trivial, we will call $\M$ a factor as in the von Neumann algebra tradition.

\subsection{Squares}
To simplify notation, in this section the capital letters $\A, \B, \C, ...$ will be used for subalgebras, not necessarily corresponding to subsystems. We denote
\[ I\bigg[\begin{array}{cc}\A\!&\M\\\C\!&\B\end{array}\bigg]_{\rho}
\equiv H(\A)+H(\B)-H(\M)-H(\C) \pl , \]
where $\A,\B,\C\subset M$ are subalgebras of a finite dimensional $C^*$ algebra $\M$ with a fixed trace $\tr$, and $\E_\M, \E_\A, \E_\B, \E_\C$ are the usual, unique, trace-preserving conditional expectations.

We recall that $\bigg[\begin{array}{cc}\A\!&\M\\ \C\!& \B\end{array}\bigg]$  is called a \emph{commuting square} \cite{popa_orthogonal_1983} if $\C \subseteq \A \cap \B \subseteq \M$, and $\E_\A \E_\B = \E_\B \E_\A = \E_\C$. When this is the case, we denote 
\begin{equation}
I(\A : \B \subseteq \M)_\rho \equiv I\bigg[\begin{array}{cc}\A\!& \M \\ \C \!& \B \end{array}\bigg]_{\rho} \pl ,
\end{equation}
with $\C$ implicitly given as $\A \cap \B$. Thanks to the pattern
 $ \bigg[\begin{array}{cc} +\!& -\\ -\!& +\end{array}\bigg]$
and Araki-Moriya's SSA \cite{araki_equilibrium_2003} inequality, we have the following observation:
\begin{lemma} (\textbf{Chain Rule}) \label{mo1} Let $\A,\B,\C,\T,\S \subset \M$ . Then
 \[ I\bigg[\begin{array}{cc}\A \!& \M \\ \C \!& \B \end{array}\bigg]
 \lel I\bigg[\begin{array}{cc} \A \!& \M \\ \T \!& \S \end{array}\bigg]+
 I\bigg[\begin{array}{cc}\T \!& \S \\ \C \!& \B \end{array}\bigg] \pl .\]
If moreover, $\bigg[\begin{array}{cc}\T \!& \S \\ \C \!& \B \end{array}\bigg]$ is a commuting square, then
  $  I\bigg[\begin{array}{cc}\A \!& \M \\ \C \!& \B \end{array}\bigg]
 \gl I\bigg[\begin{array}{cc} \A \!& \M \\ \T \!& \S \end{array}\bigg] \pl.$
\end{lemma}
\noindent This chain rule is the intuition behind a variety of inequalities, yet it follows immediately from writing out the quantities involved. Note that this does not require von Neumann algebras as the objects of study - we could for instance consider the chain rule between squares of quantum channels, states, or any other objects for which there is a meaningful notion of entropy.
\subsection{Strong Subadditivity, Uncertainty, Coherence and Asymmetry} \label{sec:connection}
In this short section, we note some connections that unify several ideas in quantum information.
\begin{definition}
Given a subalgebra $\N \subseteq \M$, we define the $\alpha$-R\'enyi asymmetry measure of relative entropy with respect to $\N$ as
\begin{equation}
D^{\N}_\alpha(\rho) \equiv \inf_{\sigma \in S_1(\N)} D_\alpha(\rho \| \sigma) \pl.
\end{equation}
In particular, we define $D^{\N}(\rho) \equiv D^{\N}_1(\rho)$.
\end{definition}
\noindent By Lemma \ref{lem:entropyforms}, $D^{\N}(\rho) = H(\E_\N(\rho)) - H(\rho)$. For a pair of subalgebras $\S, \T \subseteq \M$ forming a commuting square, $D_\alpha^{\S}(\rho) \geq D_\alpha^{\S \cap \T}(\E_\T(\rho))$ for all $1/2 \leq \alpha \leq \infty$ by data processing of sandwiched relative R\'enyi entropy \cite{muller-lennert_quantum_2013,wilde_strong_2014}. In section \ref{sec:dnprops}, we show that $D^{\N}$ is a resource monotone. $D^\N$ relates to the Holevo asymmetry as in \cite{marvian_extending_2014}, first introduced in \cite{gour_measuring_2009,vaccaro_tradeoff_2008}. $D^\N$ generalizes this notion from groups to subalgebras $\N$.
\begin{cor} \label{cor:ssa}
Let $\S, \T \subseteq \M$ form a commuting square. Let $\sigma \in S_1(\M)$ be a density such that $\E_\T(\sigma) = \sigma$. Then
\begin{align}
\label{ssa1} 
\begin{split}
I(\S : \T \subset \M)_{\rho} & \equiv - D(\E_\S(\rho) \| \E_\S(\sigma)) - D(\E_\T(\rho) \| \E_\T(\sigma)) + D(\rho \| \sigma) + D(\E_{\S \cap \T}(\rho) \| \E_{\S \cap \T}(\sigma)) \\
	& =  H(\S)_\rho + H(\T)_\rho - H(\M)_\rho - H(\S \cap \T)_\rho \geq 0 ,
\end{split}
\end{align}
with the inequality holding for all densities $\rho\in \M$ iff $\E_{\S}\circ\E_{\T}=\E_{\T}\circ\E_{\S}=\E_{\S \cap \T}$.
\end{cor}
\noindent An earlier versions of this von Neumann algebra strong subadditivity was proven by Petz in 1991 \cite{petz_certain_1991}, and a case of it was shown by Araki and Moriya in 2003 \cite{araki_equilibrium_2003}. We show a simplified form here, give a simplified proof for finite-dimensional algebras in terms of the usual data processing inequality in Appendix \ref{sec:alg}, and show that the inequality holding for all densities implies the commuting square.
\begin{rem} \label{rem:recovery}
As the proof of Corollary \ref{cor:ssa} is essentially data processing, we can strengthen the result to
\begin{equation}
I(\S : \T \subset \M)_{\rho} \geq -2 \log (F(\rho, R_{\E_\S(\rho), \E_\T} \circ \E_\T(\rho))) 
\end{equation}
in which $R$ is the Petz recovery map \cite{petz_sufficient_1986}, the universal recovery map in \cite{junge_universal_2018}, or any other recovery map $R$ for which $D^\S(\rho) - D^\S(\E_\T(\rho)) \geq -2 \log (F(\rho, R \circ \E_\T(\rho)))$ holds. This follows from equation \eqref{eq:ssapf}. We are may switch the roles of $\S$ and $\T$. Furthermore, this suggests that the idea of quantum Markov chains and approximate quantum Markov chains carry through to the subalgebraic setting. When $I(\S : \T \subset \M) = 0$, this corresponds to the existence of a perfect recovery map such that $R_{\E_\S(\rho), \E_\T} \circ \E_\T$ acts as the identity, which we may interpret as independence of $\E_\S(\rho)$ from $\E_\T(\rho)$ conditioned on any overlap contained in $\E_{\S \cap \T}(\rho)$.
\end{rem}
As some immediate consequences:
\begin{itemize}
	\item Strong subadditivity \eqref{eq:origssa} is positivity of $I(\A \C : \B \C \subset \A \B \C)$ in Corollary \ref{cor:ssa}.
	\item \label{exam:ucrmem} The uncertainty principle with memory \eqref{eq:ucrmemory}, in the case of complementary bases $\X_A$ and $\Z_A$ forming a commuting square, is positivity of $I(\X_A \B : \Z_A \B \subset \A \B)$. Via Remark \ref{rem:recovery}, refinements to strong subadditivity via recovery maps carry over to the uncertainty principle with quantum memory, implying a state-dependent tightening.
	\item Quantum coherence is intuitively the non-classicality in a given basis (see \cite{baumgratz_quantifying_2014}, or \cite{streltsov_colloquium:_2017} for a review). For an orthonormal basis $\X$, the relative entropy of coherence is defined as $C_r^{\X}(\rho) = D^\X(\rho)$ with the operational meaning of distillable coherence \cite{winter_operational_2016}. For complementary bases $\X$ and $\Z$, Corollary \ref{cor:ssa} for $I(\X : \Z \subset \A)$ implies the coherence uncertainty relation of \cite{singh_uncertainty_2016}.
	\item In terms of asymmetries, we may rewrite Corollary \ref{cor:ssa} as $D^\S(\rho)+D^\T(\rho)\ge D^{\S \cap \T}(\rho),$ a subadditivity of asymmetry. In particular, we could take $\S$ and $\T$ to be the invariant observables under the actions of groups $G_\S$ and $G_\T$. $\E_\S$ and $\E_\T$ would then be uniform averages over $G_\S$ and $G_\T$ respectively, with $\E_{\S \cap \T}$ the uniform average over both.
\end{itemize}
Replacing tensored subsystems by mutually \textit{commuting} algebras of observables has been studied before \cite{junge_connes_2011, cleve_perfect_2017, witten_aps_2018, crann_state_2019} to understand infinite-dimensional correlations and field theories. Finite-dimensional generalizations of entanglement entropy to observables or subalgebras (\cite{barnum_subsystem-independent_2004, alicki_quantum_2009, petz_algebraic_2010, derkacz_entanglement_2012, balachandran_entanglement_2013, balachandran_algebraic_2013}) have largely focused on entanglement entropy on globally pure states, and on application to systems of indistinguishable particles.

Surprisingly, positivity of $I(\S : \T \subseteq \M)$ holds even for $\S$ and $\T$ that do not commute with each other, such as mutually unbiased bases. We will further show that there are meaningful notions of side-private operations with seemingly reasonable monotones. This strongly suggests that one can define a meaningful notion of a bipartite system as a pair of algebras $\S$ and $\T$ if $[\E_\S, \E_\T] = 0$, even when $[\S, \T] \neq \{0\}$.

\subsection{Operations, Complementarity, and Symmetry}
We will generalize local operations via symmetry. Indeed, local operations are a special case of symmetries - in a bipartite system $AB$, $A'$s local operations are symmetries of the observable subalgebra $\B$. There is thus a conceptual link between symmetry and privacy. As described by Jason Crann, this is really about complementarity and commutant algebras \cite{crann_private_2016}: operations private from $\S$ are those with which all observables in $\S$ commute. Given a subgroup of the unitary group on Hilbert space $G \subseteq U(\Hilbert)$, we denote the commutant of $G$ by $G' = \{a \in \BB(\Hilbert) : a u = u a \text{ } \forall u \in G \}$. Hence $G$ is private to $G'$. For local operations, we would choose $G = U(H_A) \otimes \{\id^B\}$.

We recall that a channel $\Phi$ is \textit{covariant} with a unitary subgroup $G \subseteq U(\Hilbert)$ if $U \Phi(\rho) U^\dagger = \Phi(U \rho U^\dagger)$ for all $U \in G$ \cite{marvian_mashhad_symmetry_2015}. By definition, an asymmetry measure with respect to $G$ is non-increasing under $G$-covariant channels. In particular $\Phi \circ \E_{G'} = \E_{G'} \circ \Phi$ for any $G$-covariant channel $\Phi$. Hence $D_\alpha^{G'}(\rho) \geq D_\alpha^{G'}(\Phi(\rho))$ for any $G$-covariant $\Phi$ by data processing. We extend these ideas beyond groups to matrix algebras. We denote by $\N_\M'$ the commutant of subalgebra $\N \subset \M$ in $\M$, dropping the subscript $\M$ when it is clear from context. In the locality case, $\A$'s local operations are $U(H_A) \otimes \{\id^B\}$-covariant.
\begin{definition}
We call a channel $\Phi : S_1(\M) \rightarrow S_1(\M)$ $\T$-preserving if $\E_\T(\rho) = \E_\T(\Phi(\rho))$. In other words, $\Phi$ is within the symmetries of $\T$.

We call $\Phi$ $\T$-preserving up to isometry if $\E_\T(\Phi(\rho)) = U \E_\T(\rho) U^\dagger$ for some isometry $U : S_1(\M) \rightarrow S_1(\tilde{\M})$ for some algebra $\tilde{\M}$.
\end{definition}
\begin{definition}
Let $\Phi : S_1(\M) \rightarrow S_1(\M)$ be a quantum channel with adjoint $\Phi^\dagger : \M \rightarrow \M$. Let $\N \subset \M$ be a subalgebra. We call $\Phi$ an $(\N \subseteq \M)$-bimodule channel if $a \Phi^\dagger(b) c = \Phi^\dagger(a b c)$ for all $a,c \in \N$, and $b \in \M$. In other words, left and right multiplications of elements of $\M$ by elements of $\N$ commute with an $(\N \subseteq \M)$-bimodule channel.

We call $\Phi : S_1(\M) \rightarrow S_1(\tilde{\M})$ an $(\N \subseteq \M)$-adjusted bimodule channel if it can be written as a bimodule channel preceded by an expansion $\M \rightarrow \M \otimes \R, \N \rightarrow \N \otimes \R$ for some extra algebra $\R$, and followed by tracing out any systems left in complete mixture.
\end{definition}
\begin{rem} \label{rem:bimodcomm} \normalfont
For any $\N \subseteq \M$, $\rho \in S_1(\M)$, $b \in \M$, and $a, c \in \M$ with $\Phi$ an $(\N \subseteq \M)$-bimodule channel,
\begin{equation}
\tr(c \Phi(\rho) a b) = \tr(\Phi(\rho) a b c) = \tr(\rho \Phi^\dagger (a b c)) = \tr(\rho a \Phi^\dagger(b) c) = \tr(\Phi(c \rho a) b) \pl.
\end{equation}
Hence the bimodule property also applies in the Schr\"odinger picture. Furthermore,
\begin{equation}
\tr(a \E_\N(\Phi(\rho)) c) = \tr(\E_\N^\dagger(c a) \Phi(\rho)) = \tr(\Phi^\dagger(c a) \rho) = \tr(c a \Phi^\dagger(\rho)) = \tr(a \E_\N(\Phi(\rho)) c) \pl,
\end{equation}
so $[\Phi, \E_\N] = 0$. $D^\N(\rho) \geq D^\N(\Phi(\rho))$ for all $\rho$ by data processing for any algebra $\M \supseteq \N$. Since this holds for any $\M$, we do not explicitly denote $\M$ in the $D^\N$ notation. These properties extend easily to adjusted bimodule channels, for which $\Phi \E_\N = \E_{\tilde{\N}} \Phi$ for some $\tilde{\N} \subseteq \tilde{\M}$.
\end{rem}
\noindent The bimodule property is an analog of covariance for von Neumann algebras. Bimodule channels are free operations in the von Neumann algebra version of asymmetry.

In the bipartite setting, we revisit the connection between symmetry and privacy. If an operation is $\T$-perserving, it is invisible to any observable in $\T$. There is a challenge in defining an \textit{operational} bipartite information monotone, noting some peculiarities:
\begin{enumerate}
	\item As per Remark \ref{rem:secretshare}, $I(\S : \T \subseteq \M)$ is non-increasing when enlarging $\S$ or $\T$, but not under shrinking. It would seem that greater access to the state implies fewer resources. Indeed, $I(\CC 1 : \CC 1 \subseteq \M)_{\ketbra{\psi}} = \log |M|$ for any pure $\ket{\psi}$, the maximum attainable value, while $I(\M : \M \subseteq \M)_{\ketbra{\psi}} = 0$ has none. This reverses what we expect operationally.
	\item We might instead consider $I(\S' : \T' \subseteq \M)$, the generalized CMI between commutant algebras for $\S$ and $\T$ that form a co-commuting square. This allows lossy channels in the Heisenberg picture by performing a unitary within $\S$ or $\T$, then reducing the algebra, while $\rho^{\M}$ is invariant up to unitary. $I(\S' : \T' \subseteq \M)$ nonetheless may depend on information in $\S' \cap \T'$ that is not in $\S \T$, and it is less intuitive why we should quantify mutual information of commutants.
\end{enumerate}
\noindent We will find that $I(\S : \T \subseteq \S \T)$ has some special properties that suggest it as a canonical form. To see this, we examine complementarity in von Neumann algebras.

A quantum channel represents a physical process of open system time-evolution. As illustrated in figure \ref{fig:stine}, a quantum channel is always equivalent to unitary time-evolution in a larger system, followed by tracing out the environment. This naturally leads each channel $\Phi$ to have a \textit{complementary channel} $\Phi^c$ arising from unitary evolution with an initially pure environment, followed by tracing out the output. 
\begin{figure}
\includegraphics[width=0.8\textwidth]{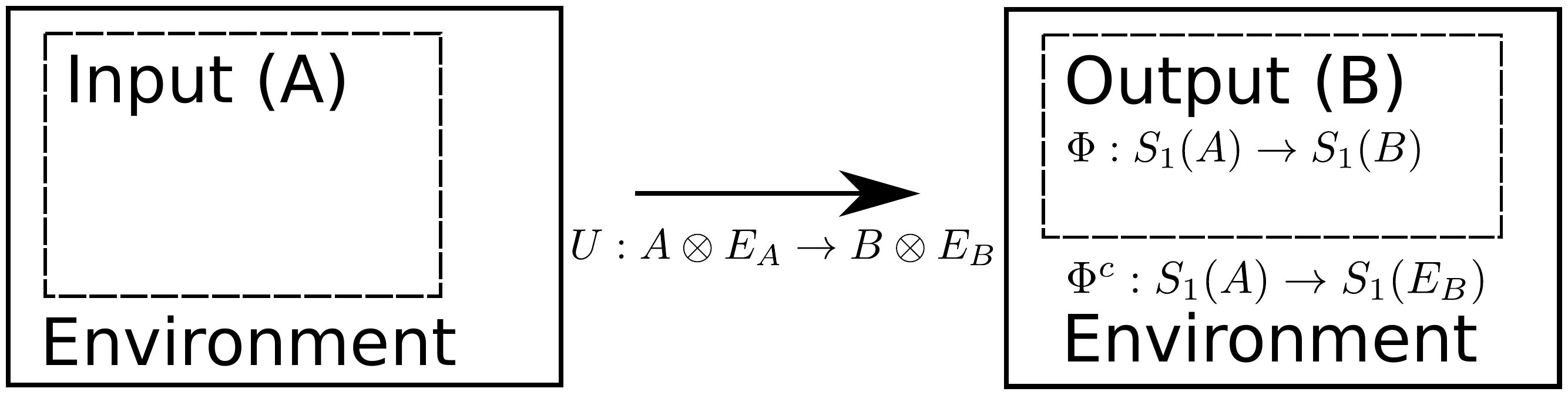}
\caption{Schematic of a quantum channel. The system and its environment undergo unitary evolution together.
\label{fig:stine}}
\end{figure}
As an immediate application of complementarity, we derive an uncertainty principle in the style of Maassen-Uffink \cite{maassen_generalized_1988} but generalized to subalgebras.
\begin{theorem}\label{B}Let $\S, \T \subseteq \A$ form a commuting square, and $B,C$ be two auxiliary systems. Then for any tripartite state $\rho^{ABC}$,
\begin{align}\label{tripa} H(\E_\S^c|\B)_{\rho}+H(\T|\C)_\rho\ge H(\S \cap \T|\C)_\rho\pl,\end{align}
where $H(\E_\S^c|\B)_{\rho} = H((\E_\S^c \otimes \id) \rho^{AB}) - H(B)_{(\E_\S^c \otimes \id)(\rho)}$. For commutative $\S$, $H(\E_\S^c|\B)_{\rho} = H(\S'|\B)_{\rho}$, and we recover the commuting square case of the original Maassen-Uffink relation. In particular, if $\R=\CC$, then $H(\R|C)_\rho=log |A|$.
\end{theorem}
\noindent A conditional expectation is a special kind of quantum channel, for which the output space corresponds exactly to a subalgebra of the full matrix algebra.

A related idea to complementarity is purification. Given a state $\rho^A$, we may purify as $\rho^A = \tr_{A^c}(\ketbra{\psi}^{M})$, where $M \supset A$, and $A^c$ is the complement of the $A$ system. Physically, purification reflects the possibility that a mixed state is the marginal of a pure, entangled state of the studied system and its environment. For a bipartite pure state $\ket{\psi}^{AB}$,
\begin{equation}
H(A)_{\ketbra{\psi}} = H(B)_{\ketbra{\psi}} \pl.
\end{equation}
Given a bipartite state $\rho^{AB}$, and purifying in $ABC$,
\begin{equation} \label{eq:condentduality}
H(B|A)_{\ketbra{\psi}} = - H(B|C)_{\ketbra{\psi}} \pl,
\end{equation}
which is a basic form of tripartite entanglement monogamy. One can easily see that we are free to permute $A$, $B$ and $C$ in the above. Given a tripartite state $\rho^{ABC} = \tr_{(ABC)^c}(\ketbra{\psi}^{ABCD})$,
\begin{equation} \label{eq:mutinfduality}
H(A|C) + H(B|C) - H(AB|C) = H(A|D) + H(B|D) - H(AB|D) \pl.
\end{equation}
For pure $\rho^{ABC}$, equation \eqref{eq:mutinfduality} implies equation \eqref{eq:condentduality}. Following Jason Crann's program \cite{crann_private_2016}, we might expect to recover this equality when replacing complements by commutants. This actually fails for $I(\S : \T \subseteq \M)$ and $I(\S' : \T' \subseteq \M)$ when traced subsystems depend on commutative parts of algebras (see Appendix \ref{sec:alg}). In contrast, recalling that $\S$ and $\T$ form a co-commuting square if $\E_{\S'} \E_{\T'} = \E_{\T'} \E_{\S'} = \E_{\S' \cap \T'}$, we have an algebraic version of equation \eqref{eq:mutinfduality}.
\begin{theorem} \label{thm:dual}
Let $\S, \T \subset \M$ be subalgebras forming commuting and co-commuting squares, for which $\M$ is a factor. Let $\rho^{\S \T} = \E_{\S \T}(\ketbra{\psi}^{M})$. Then
\begin{equation}
I(\S : \T \subset \S \T)_{\rho} 
	= I(\S' : \T' \subset \S' \T')_{\ketbra{\psi}} \pl.
\end{equation}
\end{theorem}
\noindent Theorem \ref{thm:dual} is a good hint as to the likely right form in the algebraic setting. We now examine the operational resource theory.
\begin{definition}
Let $\S, \T \subseteq \S \T$ form a commuting and co-commuting square. Then we define $I(\S : \T) \equiv I(\S : \T \subseteq \S \T)$ as an entropic measure. In particular, we interpret the following:
\begin{enumerate}
	\item $\s \equiv \S \cap (\S \cap \T)'$ as the algebra of observables accessible to one party, and $\t \equiv \T \cap (\S \cap \T)'$ as that of observables for the other.
	\item $\S \cap \T$ as a frozen memory, which the two parties may not disturb.
\end{enumerate}
We also define the following notation:
\begin{equation} \label{eq:condnotate}
I(\s : \t | \S \cap \T) \equiv I(\S : \T) \pl.
\end{equation}
\end{definition}
\noindent Here we see $\s$ and $\t$ as the two parties, analogous to $A$ and $B$ in the tensor setting for $I(A : B |C)$. $\S \cap \T$ plays the role of the conditioning system $C$. Because $I(A : B | C)$ is not monotonic in either direction under channels on $C$, even in the tensor model there is some indication that a resource theory should not allow the two parties to modify $C$. For tensor factors, we can immediately separate $\S \cap \T$ from $\S$ and $\T$ as a 3rd party. For arbitrary von Neumann algebras, we may have commuting subalgebras that appear in $\S$, $\T$ and $\S \cap \T$. We construct operations as follows:
\begin{definition} (\textbf{State-Modifying Individual Operations})
Let $\Phi : S_1(\S \T) \rightarrow S_1(\tilde{S} \tilde{\T})$ be a quantum channel that is:
\begin{enumerate}
	\item An $\S$-adjusted bimodule.
	\item $\T$-preserving up to isometry.
\end{enumerate}
Then we call $\Phi$ a state-modifying or Schr\"odinger picture $\S$-operation (or $\S$-op for short). We define state-modifying $\T$-operations analogously.
\end{definition}
\noindent We may interpret the $\S$-adjusted bimodule property as implying that the transformation $\rho \rightarrow \Phi(\rho)$ still looks like itself from $\S$'s perspective, in that $\E_\S(\rho) \rightarrow \Phi(\E_\S(\rho))$. Meanwhile, $\T$-preservation prevents $\S$-operations from transmitting any information to $\T$. We may relax preservation by an isometry, representing ways $\Phi$ may change the way in which $\T$ embeds in the larger algebra, as long as these don't change the entropy of $\E_\T$. Proposition \ref{prop:bimod} gives a full description of $\S$-bimodule channels that only enlarge $\T$, implying a complete characterization of state-modifying $\S$-operations.

\begin{theorem} \label{thm:mono}
Under an $\S$-operation for which $\S \rightarrow \tilde{\S}$, $\T \rightarrow \tilde{\T}$, and $\rho \rightarrow \tilde{\rho}$, 
\begin{equation}
I(\S : \T)_\rho \geq I(\tilde{\S} : \tilde{\T})_{\tilde{\rho}} \geq 0 \pl.
\end{equation}
Analogously, $I$ is also non-increasing and positive under any sequence of $\S$-ops and $\T$-ops.
\end{theorem}
\noindent The proof of this Theorem appears in Appendix \ref{sec:opproofs}.
\begin{rem} \normalfont
As in Remark \ref{rem:recovery}, given the Petz recovery map $R_{\E_{\S}(\rho), \Phi}$ or a substitutable recovery map and under a state-modifying $\S$-operation $\Phi$,
\begin{equation}
I(\S : \T)_\rho - I(\S : \T)_{\Phi(\rho)} \geq -2 \log (F(\E_{\S \T}(\rho), R_{\E_{\S}(\rho), \Phi} \circ \Phi(\E_{\S \T}(\rho)))) .
\end{equation}
\end{rem}
\noindent We also have a notion of algebra-modifying individual operations, which allow us to change $\S$ and $\T$. As these are more technical, we introduce and explain them in Appendix \ref{sec:opproofs}. We give an abridged summary here.
\begin{prop} \normalfont \label{prop:extrafree}
Let $\S, \T$ form a commuting square. Then $I(\S : \T)_\rho$ is non-increasing under:
\begin{enumerate}
	\item $\rho \rightarrow U \rho U^\dagger$, and $a \rightarrow U a U^\dagger$ for each $a \in \S \T$. In this case, the transformation on the algebra side cancels that on the state side.
	\item $\rho \rightarrow U \rho U^\dagger$, where $U$ is an isometry that is an adjusted bimodule for $\S, \T, \S \T$. This is effectively a global change of coordinates, which has no effect on entropies.
	\item $\S \rightarrow \tilde{\S}$, where $\tilde{\S} \subseteq \S$ as an algebra, and $\tilde{\S}, \T$ form a commuting square. This is $\S$ or $\T$ dropping access to some observables.
	\item $\S \rightarrow \S \R$, where $\R \subseteq \T$, and $\S \R, \T$ form a commuting square. This allows adding elements of $\S$ or $\T$ to the immutable shared memory, $\S \cap \T$. In this case of non-commutative $\R$, these are frozen and effectively hidden from $\s$ or $\t$. If $\R$ is commutative, then the immutability of $\S \cap \T$ does not necessarily prevent access. Relatedly, the traditional CMI, $I(A : B | C)$, allows some transfer of classical information from $A$ or $B$ to $C$.
	\item Some algebraic changes, such as enlarging $\S$ or $\T$ in a way that doesn't enlarge $\S \cap \T$, are $I$-non-increasing for particular states. We may apply a channel to $\rho$ that ensures the state has such a form while simultaneously changing the algebras, creating a free transformation for all states. We describe this in detail in Lemma \ref{lem:swap}. We may interpret this kind of transformation as a change in assumptions about the environment - a more mixed $\rho$ with larger $\S$ or $\T$ may have the same expectations in $\S$ and in $\T$. \label{extrafree5}
\end{enumerate}
\end{prop}
\noindent In this section, we have constructed a resource theory (see \cite{brandao_reversible_2015}) to understand individual operations for bipartite observable algebras that may not mutually commute. This is analogous to the resource theory of local operations (LO) under which mutual information does not increase, but we have extended it to potentially interacting or even overlapping systems. Our next step is to construct a version analogous to entanglement, which quantifies non-classical effects.
\subsection{Convex, Non-Classical Monotones} \label{sec:measures}
$I(\S : \T)$ reduces to $I(A : B)$ in the case of non-overlapping tensor-separated subsystems, which counts classical as well as quantum correlations. Inspired by the theory of entanglement monotones \cite{horodecki_quantum_2009}, we generalize ``squashed" and convex roof entanglement measures.
\begin{definition}
Let $\S,\T$ form a commuting square such that $\S \cap \T = \CC 1$. We define the squashed mutual information as
\begin{equation} \label{eq:isqcomm}
I_{sq}(\S : \T)_\rho \equiv \frac{1}{2}\inf_{\tilde{\S}, \tilde{\T}; \tilde{\rho} \in S_1(\tilde{\S} \tilde{\T})} I(\tilde{\s} : \tilde{\t} | \tilde{\S} \cap \tilde{\T})_{\tilde{\rho}} \pl,
\end{equation}
where $\tilde{\S}, \tilde{\T} \subseteq \tilde{\S} \tilde{\T}$ must form a commuting square, and we must still have $\tilde{\s} = \s$, and $\tilde{\t} = \t$ with the expectations of observables in $\s \t$ unchanged.
\end{definition}
\noindent This generalizes the squashed entanglement of \cite{christandl_squashed_2004}. If $\S \cap \T$ were to contain any non-commutative subsystems, it is ambiguous whether $\tilde{\S} \cap \tilde{\T}$ should retain the existing $\S \cap \T$ as a subalgebra, or should be free to replace parts of it that wouldn't affect $\s$ or $\t$. The right answer to this question should result from operational and physical applications, and it may depend on the specific physical setting.
\begin{rem} \label{rem:sqfactor} \normalfont
Were $\tilde{\S} \cap \tilde{\T}$ to contain a non-trivial center, it would appear in $\tilde{\s}$ and $\tilde{\t}$, violating the minimization constraints. Hence we may assume that $\tilde{\S} \cap \tilde{\T}$ is a factor. We may further assume the form $\tilde{\M} \cong \mathbb{M}_{|\S \T|} \otimes (\tilde{\S} \cap \tilde{\T})$, because anything outside of this will be effectively traced out by $\E_{\tilde{\S} \tilde{\T}}$.
\end{rem}
\begin{definition}
Let $\S,\T$ form a commuting square with density $\rho$ such that $\S \cap \T = \CC 1$. We define the following convex roof measure,
\begin{equation}
I_{conv}(\S : \T)_\rho \equiv \frac{1}{2} \inf_{\{(p_x, \rho_x) \}} \sum_{x} p_x I(\S : \T)_{\rho_x} \pl ,
\end{equation}
where $\{p_x\}$ forms a probability distribution such that $\rho = \sum_x p_x \rho_x$, and $\rho_x \in S_1(M)$ for all $x$.
\end{definition}
\begin{rem} \label{rem:entl} \normalfont
For a bipartite state $\rho^{AB}$ on tensored factors $A \otimes B$, $I_{sq}(\A : \B)_\rho = E_{sq}(A : B)_\rho$ as in equation \eqref{eq:esquash}, and $I_{conv}(\A  : \B) = E_F(A : B)_\rho$, where $E_F$ is the entanglement of formation defined in \cite{bennett_mixed-state_1996}. For a bipartite pure state $\ket{\psi}^{AB}$, both $I_{sq}$ and $I_{conv}$ reduce to the entanglement entropy
\begin{equation}
\frac{1}{2} I(\A : \B)_{\ketbra{\psi}} = \frac{1}{2} I(A : B)_{\ketbra{\psi}} = H(A)_{\ketbra{\psi}} = H(B)_{\ketbra{\psi}} \pl .
\end{equation}
\end{rem}
\begin{cor} \label{cor:sqdual}
Let $\S, \T \subseteq \S \T \subseteq \M$ form a commuting square, $\M \cong \mathbb{M}_n$ a minimal such factor, $\rho = \E_{\S \T}(\rho)$, and $\S \cap \T = \CC 1$. We may rewrite
\begin{equation}
I_{sq}(\S : \T)_\rho
= \inf_{\ket{\psi} \in \S \T \C \D} \frac{1}{2} I(\S_\M' \D: \T_\M' \D)_{\ketbra{\psi}} \pl,
\end{equation}
where $\ket{\psi}$ purifies $\rho$ in $\M \C \D$, and $\C$ and $\D$ are factors.
\end{cor}
\begin{proof}
Via remark \ref{rem:sqfactor}, we may assume that the infimum is taken with respect to $\tilde{\rho}$ in $\S \T \otimes \C$ for some tensor extension $\C$. We purify $\tilde{\rho}^{\M \C}$ as $\ket{\psi}^{\M \C \D}$ and apply Theorem \ref{thm:dual}.
\end{proof}
\begin{cor} \label{cor:sqpure}
Purifying extensions are insufficient for the infimum in $I_{sq}$ or $E_{sq}$.
\end{cor}
\begin{proof}
Consider the traditional case of $E_{sq}(A:B)_\rho$ on $A \otimes B$, where $\rho$ is a classically correlated (separable) density. Were the achieving extension to be pure, any purifying $D$ would itself be in a pure state. By Corollary \ref{cor:sqdual} and Remark \ref{rem:entl}, we rewrite this as 
\begin{equation}
\frac{1}{2} I(\A \D : \B \D \subset \A \B \D)_{\rho^{AB} \otimes \ketbra{\phi}^D} = \frac{1}{2} I(A : B)_{\rho^{AB}} \pl .
\end{equation}
Since $\rho^{AB}$ was assumed to be classically correlated, it must have non-zero mutual information. Were this to achieve the infimum, it would show that a separable state has non-zero squashed entanglement, which would contradict its faithfulness proven in \cite{brandao_faithful_2011, li_squashed_2018}.
\end{proof}
\noindent Corollary \ref{cor:sqpure} also suggests that we could interpret the squashed entanglement as making the most conservative possible assumptions about the environment, given that any mixed state may have some information copied in the environment. $I_{conv}$ then has a similar interpretation with a classical eavesdropper. A pure state in a factor is necessarily private. Here we enumerate some basic properties of these measures, as proven in Appendix \ref{sec:sqcproofs}.
\begin{prop} \normalfont
For $\S, \T$ forming a commuting square with $\S \cap \T = \CC1$, and density $\rho$:
\begin{enumerate}
	\item \label{entlprop1} $I_{sq}$ and $I_{conv}$ are convex in $\rho$. Hence both are maximized on pure states.
	\item \label{entlprop2} If $\E_{\S \T}(\rho)$ is pure in $S_1(\S \T)$, then $I_{conv}(\S : \T)_\rho = I_{sq}(\S : \T)_\rho = \frac{1}{2}I(\S : \T)_{\rho}$.
	\item \label{entlprop3} If $\S \T$ is commutative, then $I_{sq}(\S : \T)_\rho = I_{conv}(\S : \T)_\rho = 0$.
	\item \label{entlprop4} $I_{sq}$ and $I_{conv}$ are trace distance continuous in $\rho$ (see Lemma \ref{lem:contiuoussq}).
	\item \label{entlprop5} $I_{sq}(\S : \T)_\rho \leq I_{conv}(\S : \T)_\rho$.
\end{enumerate}
The following properties apply to $I_{sq}$ but not necessarily to $I_{conv}$.
\begin{enumerate}
	\setcounter{enumi}{5}
	\item \label{entlprop6} Let $\M = \otimes_{i=1}^n \M_i$ with $\S_i \T_i = \M_i$ in co-commuting square for each $i \in 1...n$, so that $\S = \otimes_{i=1}^n \S_i$, and $\T = \otimes_{i=1}^n \T_i$. Then
	\begin{equation}
		\sum_i I_{sq}(\S_i : \T_i)_{\E_{\M_i}(\rho)} \leq I_{sq}(\S : \T)_\rho.
	\end{equation}
	If $\rho = \otimes_{i=1}^n \rho_i$ such that $\rho_i \in S_1(M_i)$, then equality is achieved. This property is analogous to and generalizes monogamy of squashed entanglement.
	\item \label{entlprop7} Let $\S = \S_1 \otimes \S_2$, and $\T = \T_1 \otimes \T_2$. Then $I_{sq}(\S : \T)_{\rho} \geq I_{sq}(\S_1 : \T_1)_{\rho^{\S_1 \T_1}}$.
	\item \label{entlprop8} Let $\A, \B \subset \S \T$ be factors such that $\S \subseteq \A$, $\T \subseteq \B$, $\A \cap \B = \CC 1$, and $I_{sq}(\S : \T)_\rho \leq \epsilon$. Then there exists some separable $\sigma$ such that $\|\E_{\S \T}(\rho) - \E_{\S \T}(\sigma) \|_1 \leq  3.1 |B| \sqrt[4]{\epsilon}$, and some separable $\eta$ such that $\|\E_{\S \T}(\rho) - \E_{\S \T}(\eta) \|_2 \leq 12 \sqrt{\epsilon}$. We similarly have closeness to highly extendible states (Lemma \ref{lem:sqfaith}). This slightly generalizes Brandao, Christandl and Yard's faithfulness result in \cite{brandao_faithful_2011}, and Li and Winter's results in \cite{li_squashed_2018}.
\end{enumerate}
\end{prop}
Extending monogamy of entanglement to $I_{sq}$ to non-commuting algebras is generally impossible, as shown by the following counter-example:
\begin{exam} \label{exam:mononogo}
The superadditivity/monogamy of $I_{sq}$ does not necessarily generalize to situations in which $\S = \cup_{i=1}^n \S_i$ or $\T = \cup_{i=1}^n \T_i$ when they are not tensor products. For example, consider a high-dimensional factor $M$ with $m_b$ mutually unbiased bases. Let the system be in state $\ketbra{\psi}$, which is prepared in a particular mutually unbiased basis. Let $i$ and $j$ index the rest of the mutually unbiased bases. Then
\begin{equation}
\sum_{i,j} I_{sq}(\S_i : \T_j)_{\ketbra{\psi}} = \frac{m_b \log |M|}{2}.
\end{equation}
This can be much larger than $\log |M| / 2$, which we will see via Theorem \ref{thm:yesquantumsq} is the largest possible value of $I_{sq}(\S, \T)_\rho$ for any $\S, \T \subset \M$ as a commuting square.
\end{exam}
This lack of monogamy in $I_{sq}$ between bases probably relates to Li and Winter's conclusion that ``while entanglement of formation is essentially about the distance from separable states, squashed entanglement is about the distance from highly extendible states \cite{li_squashed_2018}." The usual notion of extendibility is based on symmetric extendibility, which is natural for mutually commuting subalgebras. In contrast, one may find a state $\tilde{\sigma}$ such that for each $i \in 1...k$ and algebras $\S$ and $\T_i$, $\E_{\S \T_i} = \rho$, but the state is not necessarily symmetrically extendible if $\T_i$ and $\T_j$ do not commute. Similarly, the main barrier to proving faithfulness when subalgebras do not embed in non-overlapping tensor factors may be the question of defining a notion of ``separable" in this context, and thus having something to faithfully indicate.
\begin{theorem} \label{thm:yesquantumsq}
If $\M$ is a finite-dimensional noncommutative von Neumann algebra, then there exist subalgebras $\S, \T \subset \M$ in co-commuting square and a pure state $\ketbra{\psi} \in S_1(\S \T)$ such that $I_{sq}(\S : \T)_{\ketbra{\psi}} > 0$. In particular, if $\M_0$ is the largest matrix block in $\M$, then 
\begin{equation}
I_{sq}^{max}(\M) \equiv \max_{\S, \T, \rho} I_{sq}(\S : \T)_\rho = \frac{\log |M_0|}{2}
\end{equation}
where the maximization is over $\S, \T \subset \M$ in commuting square, and $\rho \in S_1(M)$. This optimum is achieved by choosing $\S$ and $\T$ corresponding to two mutually unbiased bases of $\M_0$, and preparing $\rho$ as a pure state in a third. $I^{max}_{conv}(\M) = I^{max}_{sq}(\M)$.
\end{theorem}
\begin{rem} \label{rem:vsentl} \normalfont
For the usual subsystem-to-subsystem squashed entanglement on a factor $|M|$, if $|M|$ is a square, we can divide the system into two subsystems of equal dimension with $\frac{1}{2} \log |M|$ ebits of entanglement entropy at maximum. If $|M|$ is not a square, the maximum entanglement is lower. Hilbert spaces of dimension 2 or 3 do not factor into subsystems, yet they still support $I_{sq} > 0$. These systems cannot contain non-locality or entanglement, yet they are non-classical.

In contrast, we can always achieve the maximum of \ref{thm:yesquantumsq} with mutually unbiased bases. Together with property \ref{entlprop3}, Theorem \ref{thm:yesquantumsq} implies that a system supports non-zero $I_{sq}, I_{conv} > 0$ if and only if it contains non-commuting observables, a common notion of quantumness.
\end{rem}
\noindent Finally, we consider operations under which $I_{sq}$ and $I_{conv}$ are monotonic:
\begin{cor} \label{cor:sqindivops} (\textbf{Individual Operations})
$I_{sq}$ and $I_{conv}$ are non-increasing under individual $\S$ and $\T$-operations.
\end{cor}
\noindent We also find a class of transformations that allows us to add elements to the algebras $\s$ and $\t$ in return for mixing the state.
\begin{theorem} (\textbf{Covariant Averaging with Algebra Replacement}) \label{thm:entlcovar}
Let $\S, \T \subseteq \S \T$ form a commuting square. Let $\E_\R(\rho) = \int_G U \rho U^\dagger d\mu(U)$, an integration over the unitary subgroup $G$ with probability measure $\mu$ (not necessarily Haar). Let each $U$ with non-zero measure commute with $\E_\S, \E_\T, \E_{\S \T}$, and $\E_{\S \cap \T}$ (in the language of asymmetry, these conditional expectations are covariant). Let $\tilde{\S}, \tilde{\T}$ also form a commuting square within $\S \T = \tilde{\S} \tilde{\T}$, such that $\tilde{\S} \cap \R = \S \cap \R$, $\tilde{\T} \cap \R = \T \cap \R$, and $\tilde{\S} \cap \tilde{\T} \cap \R = \S \cap \T \cap \R$. Then
\begin{equation}
I_{*}(\S : \T)_\rho \leq I_{*}(\tilde{\S} : \tilde{\T})_{\E_\R(\rho)},
\end{equation}
where $I_{*} \in \{I_{sq}, I_{conv}\}$.
\end{theorem}
\noindent The primary physical use of Theorem \ref{thm:entlcovar} is when the averaging has no effect on the state. Rather, when observers know that a given state is inside a smaller subalgebra than their joint algebra, they can change assumptions about the environment and accessible observables. This is the non-local version of transformation \ref{extrafree5} in Proposition \ref{prop:extrafree}.
\subsection{Entanglement and Uncertainty}
As per Remark \ref{rem:vsentl}, $I_{sq}$ can be non-zero even in systems of dimension 2 or 3, which cannot contain entanglement. We do not see $I_{sq}$ as an entanglement measure, but as quantifying a broader form of nonlocality that includes as its other well-known case a form of quantum uncertainty. In particular,
\begin{equation} \label{eq:ucrisq}
I_{sq}(\X : \Z)_{\ket{\uparrow_Y}} = 1/2 \pl,
\end{equation}
where $\X, \Z$ are the Pauli bases of a qubit, and $\ket{\uparrow_Y} = \frac{1}{\sqrt{2}}(\ket{0} + i \ket{1})$ is a $Y$-eigenstate. It is possible to transform two copies of this to an entangled sttate.
\begin{cor} (\textbf{EPR $\leftrightarrow$ 2UCR}) \label{cor:final}
Let $\S = \braket{Z_A, Z_B}$, and $\T = \braket{X_A, X_B}$. Let $\rho = \ket{\uparrow_Y \uparrow_Y}^{AB}$. Then there exists a transformation under which $I_{sq}$ and $I_{conv}$ are non-increasing that converts this configuration to $\S = \braket{X_A, Z_A Z_B}$, $\T = \braket{X_A X_B, Z_B}$. This configuration is equivalent to $\frac{1}{\sqrt{2}} (\ket{0-} + i \ket{1+})$, with $\S = \A$ and $\T = \B$, The transformation is reversible.
\end{cor}
The algebras $\S = \braket{Z_A, Z_B}, \T = \braket{X_A, X_B}$ immediately allow extraction of correlation from the state $\ket{\uparrow_Y \uparrow_Y}^{AB}$, as it is easy to check that the values of $Z_A Z_B$ and $X_A X_B$ are correlated. Corollary \ref{cor:final} implies we can extend the resource theory of individual operations to one in which EPR $\leftrightarrow$ 2UCR, where ``EPR" refers to a single Bell or Einstein-Podolsky-Rosen state, a maximally entangled pair of qubits, and $UCR$ refers to the configuration of equation \eqref{eq:ucrisq}.

Entanglement depends on one's choice of how to decompose the Hilbert space, arguably relying on a priori spatial structure. In particular, one can factor a 2-qubit maximally entangled basis into a product of two (non-local) qubit bases, in which ordinary product states would look entangled, and ordinary entangled states look like products. Relatedly, entanglement-like correlations between aspects of a single particle show teleportation and Bell violations \cite{boschi_experimental_1998, michler_experiments_2000}. Hence even for tensor products, one can find non-separability within local spaces.

In the absence of spatial structure, one may define an $n$-qubit system as any $n$ mutually commuting pairs of anticommuting, norm-one observables \cite{chao_overlapping_2017}. The setting we consider of algebras $\S = \braket{Z_A, Z_B}, \T = \braket{X_A, X_B}$ flips this around: each party accesses a pair of commuting observables, which anticommute with the other party's. It therefore differs from any tensor product of systems. It nonetheless supports a positive analog of mutual information, some generalized measures of non-classical correlation, and a consistent resource theory in which we may convert between uncertainty-based and entanglement-based non-classicality.

\section{Conclusions and Discussion}
The uncertainty relations derived in this work do not completely subsume the uncertainty relation with memory, which does not require a commuting square. We refer the interested reader to our results in \cite{gao_uncertainty_2018}, which extend both beyond conditional expectations and beyond commuting squares. In the absence of a commuting square, it is unclear how to interpret the generalized CMI in the context of correlations or resources. Negativity of generalized CMI may reflect a fundamental incompatibility, or leave open the question of whether new physical ideas will explain these cases.

The most immediate consequence of the connections shown herein is that we can transfer tools in inequalities between the settings of uncertainty, asymmetry, coherence, corrleations and others. We can also explicitly model systems not falling entirely in one known setting, such as particular cases in which two observers may each control access to a group of transformations and corresponding observables. On a more fundamental and conceptual level, these mathematical links suggest connections between physical settings. Entanglement and uncertainty are two of the most famous aspects of quantum mechanics. We see a strong connection between them and linking both to the concept of symmetry.

For a system of independent random variables, entropy is an extensive property, scaling with the number of subsystems. Fully correlated but classical random variables make it an intensive property, in which the entropy of the system is equal to that of each subsystem. Quantum correlations allow the entropy of the whole to be less than the entropy of each constituent. This connects with EPR's \cite{einstein_can_1935} claimed paradox - the local theories are individually incomplete, thereby making less deterministic predictions than would be possible with access to the whole. Relatedly, the privacy of quanta enforced by the no-cloning theorem \cite{wootters_single_1982}, existence of purification and role of complementary channels suggest that quantum phenomena fundamentally involve a system's role as part of a more complete whole. The aforementioned quantum entropic phenomenon also appears for mutually unbiased bases. We might think of observables in a single basis as another kind of incomplete theory. The quantumness we quantify is the incompleteness of partial algebras, relative to the completeness of the joint.

\section{Acknowledgments}
We thank Mark M. Wilde for comments on an earlier version of this work, and Reinhard Werner for helpful discussions. LG acknowledges support from Trjitzinsky Fellowship. LG and NL acknowledge support from NSF grants DMS-1700168 DMS-1800872. NL was supported by Graduate Research Fellowship Program DMS-1144245.  MJ was partially supported by NSF grant DMS-1501103. We are grateful to Institut Henri Poincar\'e's hospitality during participation in the trimester ``Analysis in Quantum Information Theory''.

\bibliographystyle{abbrv}
\bibliography{ssa_ucr_entl}

\appendices

\section{Methods and Longer Proofs}
\subsection{Observable Algebras} \label{sec:alg}
Quantum information theory traditionally focuses on Schr\"odinger picture quantum mechanics, in which processes affect density matrices. We will however often find it illuminating to consider the Heisenberg picture, in which processes affect observables. Traditional quantum mechanics models observables as Hermetian matrices, of which each eigenvalue may specify a measurement outcome, observation of which projects the state into its corresponding eigenstate(s). From the perspective of information theory, an observable is fully characterized by its distinct eigenvectors. Information theory describes the statistics of messages rather than their content, so we will rarely need to consider the literal eigenvalues of observables or their interpretation as physical quantities.

We will primarily be interested in the algebras generated by quantum observables as von Neumann subalgebras of matrix algebras. In a von Neumann algebra $\M$, we will allow and assume closure under transformations of the following forms:
\begin{enumerate}
	\item Linear combinations with real coefficients: $a X + b Y \in \M$ $\forall X,Y \in \M$ and $a,b \in \RR$.
	\item Re-scaling of distinct entries in the diagonal basis.
	\item Composition of elements.
	\item Hermetian conjugates.
\end{enumerate}
A basic fact from von Neumann algebra theory is that in finite-dimension, a von Neumann algebra is always a block diagonal matrix algebra in some basis. In particular, these have the form
\begin{equation} \label{eq:vnaform}
\N = \oplus_i (\N_{n_i} \otimes \CC1_{m_i}),
\end{equation}
where $i$ indexes the diagonal blocks, and $n_i$ and $m_i$ quantify the dimensions of the block subspace and a potentially traced subspace. The conditional expectation as given in the Schr\"odinger picture and its commutant are given by
\begin{equation} \label{eq:condexpform}
\begin{split}
\E_\N(\rho) & = \oplus_i(\tr_{m_i}(P_i \rho P_i) \otimes \id_{m_i} / m_i) \\
\E_{\N'}(\rho) & = \oplus_i(\id_{n_i} / n_i \otimes \tr_{n_i}(P_i \rho P_i)),
\end{split}
\end{equation}
where $P_i \rho P_i / \tr(P_i \rho P_i)$ is a density in an $(n_i \times n_i) \otimes (m_i \times m_i)$-dimensional matrix space, and $\tr_{m_i}$ traces out the $m_i$-dimensional space. Because this conditional expectation is both trace-preserving and unital, $\sum_i n_i m_i = |M|$. We are particularly interested in the relationship between $\E_{\N'}$ and $\E_{\N}^c$, the respective commutant and complement of $\E_\N$ as a quantum channel.

Let $A \subseteq M$. The channel $\rho^{M} \rightarrow \tr_{A}(\rho) \otimes \id_A/|A|$ traces out $A$, replacing it by complete mixture. There is a simple and minimal Stinespring dilation $U$ given by
\begin{equation} \label{eq:trstine}
U\rho U = \text{SWAP}(A, E_1) \Big ( \rho \otimes \frac{1}{|A|} \sum_{i,j} \ketbra{i,j}^{E_1 E_2} \Big ).
\end{equation}
The environment is left with $\rho^A$, and an extra system that is maximally entangled with the completely mixed output in $A$. We may Stinespring dilate a conditional expectation as a direct sum of these,
\begin{equation} \label{eq:condexpstine}
U_{\E_\N} \rho U_{\E_\N}^\dagger = \oplus_{i,j} (\id_{n_i} \otimes U_i ) (P_i \rho P_j) (\id_{n_j} \otimes U_j^\dagger),
\end{equation}
where $(\id_{n_i} \otimes U_i)$ dilates $\rho \rightarrow \tr_{m_i}(\rho) \otimes \id_{m_i} / m_i$ via equation \eqref{eq:trstine}. This naturally yields the same block diagonal structure on $E$. We thus calculate
\begin{equation} \label{eq:condexpcomp}
\E_{\N}^c = \oplus_i (\id_{m_i} / m_i \otimes \tr_{n_i}(P_i \rho P_i)) \pl ,
\end{equation}
The only difference between $\E_{\N'}$ and $\E_\N^c$ is the blockwise dimension filled by complete mixture.

A common motivation for replacing tensored Hilbert spaces by von Neumann algebras is to study field theories and other infinite-dimensional settings in which tensor factors are invalid \cite{witten_aps_2018}. These theories often have divergent entanglement entropy for subregions, or even a phenomenon known as embezzlement \cite{cleve_perfect_2017}, in which parties may extract EPR pairs from a pre-shared catalyst or even the vacuum. It is crucial to our conclusion, however, that there be no way to draw entanglement from the algebraic structure of the theory. Hence we focus on the basic setting of finite-dimensional Hilbert spaces. We have taken great care to ensure that the entanglement and the non-classicality of uncertainty can really only come from each other. In future work, one may consider if in some settings that forbid tensor decomposition, it is even possible to distinguish entanglement from the uncertainty-like non-classical entropy studied here.

\begin{lemma} \label{lem:condexplog}
Let $\N \subset \M$ with conditional expectation $\E_\N$. Then for any density $\rho \in \M$,
\begin{equation}
\E_\N \log \E_\N(\rho) = \log \E_\N(\rho).
\end{equation}
\end{lemma}
\begin{proof}
We use the explicit form in equation \eqref{eq:condexpform}.
\begin{equation}
\begin{split}
\log \E_\N(\rho) & = \log (\oplus_i (\tr_{m_i} (P_i \rho P_i) \otimes \hat{1}_{m_i} / m_i)) \\
& = \oplus_i (\log (\tr_{m_i} (P_i \rho P_i)  / m_i ) \otimes \id_{m_i}) .
\end{split}
\end{equation}
If we apply $\E_\N$ again, we will trace the $\id_{m_i}$ in each $i$-block, yielding an extra $m_i$ factor, which is canceled by the normalization.
\end{proof}

\begin{lemma} \label{lem:entropyforms}
For any subalgebra $\N \subset \M$ and any $\rho \in M$, $H(\E_\N(\rho)) - H(\rho) = D(\rho \| \E_\N(\rho))$.
\end{lemma}
\begin{proof}
First, we expand
\begin{equation}
\begin{split}
D(\rho \| \E_\N(\rho)) = \tr(\rho \log \rho - \rho \log \E_\N(\rho)) \pl .
\end{split}
\end{equation}
Applying Lemma \ref{lem:condexplog}, $\tr(\rho \log \E_\N(\rho)) = \tr(\E_\N(\rho) \log \E_\N(\rho))$. The definition of von Neumann entropy then completes the proof.
\end{proof}
\begin{lemma}
Let $\N \subseteq \M$. Then
\begin{equation}
\inf_{\sigma} D(\rho \| \sigma) = D(\rho \| \E_\N(\rho))
\end{equation}
if the infimum is restricted to run over densities.
\end{lemma}
\begin{proof}
The optimal value is always attained by condition expectation because for any $\si\in S(\N)$,
\begin{align*}D(\rho||\si)&=\tau(\rho\log \rho -\rho\log \si)=\tau(\rho\log \rho)-\tau(\E_\N(\rho)\log \si)\\&=\tau(\rho\log \rho-\E_\N(\rho)\log \E_\N(\rho))-\tau(\E_\N(\rho)\log \si-\E_\N(\rho)\log \E_\N(\rho))\\&=D(\rho||\E_\N(\rho))+D(\si||\E_\N(\rho))\\&\ge D(\rho||\E_\N(\rho))\pl,\end{align*}
where in the last step we use the non-negativity of $D(\cdot||\cdot)$.
\end{proof}
\begin{lemma} \label{lem:petzrecov}
Let $\N \subset \M$ be a subalgebra. Denote by
\begin{equation} \label{eq:petzformula}
R_{\id/|M|, \E_\N^c}(\sigma) \equiv \Big ( \frac{\id}{|M|} \Big )^{1/2} \E_{\N}^{c \dagger} \Big ( \Big ( \E_\N^c \Big ( \frac{\id}{|M|} \Big ) \Big )^{-1/2}
	 \sigma \Big ( \E_\N^c \Big ( \frac{\id}{|M|} \Big ) \Big )^{-1/2} \Big ) \Big ( \frac{\id}{|M|} \Big )^{1/2}
\end{equation}
the Petz recovery map for $\E_\N^c$ acting on state $\sigma$, with default state $\id/|M|$. Then
\begin{equation}
R_{\id/|M|, \E_\N^c} \circ \E_\N^c = \E_{\N'} \pl,
\end{equation}
and
\begin{equation} \label{eq:petzform}
R_{\id/|M|, \E_\N^c}(\sigma) =  \oplus_i (m_i 1_{n_i} / n_i \otimes \tr_{m_i}(P_i \sigma P_i)) \pl.
\end{equation}
$\E_\N^c$ recovers $R_{\id/|M|, \E_\N^c}$ for default state $\E_\N^c(\id/|M|)$.
\end{lemma}
\begin{proof}
First, we explicitly calculate the complement on $\id / |M|$,
\begin{equation}
\E_\N^c(\id/|M|) = \oplus_i \frac{n_i \id_{m_i^2}}{|M| m_i} \pl.
\end{equation}
This is easy to invert, so
\begin{equation}
(\E_\N^c(\id/|M|))^{-1/2} = \oplus_i \sqrt{\frac{|M| m_i}{n_i}} \id_{m_i^2} \pl.
\end{equation}
The effect of the two factors of this is to change the normalization of each block by $|M| m_i / n_i$. Similarly, the two factors of $(\id/|M|)^{1/2}$ adjust the overall normalization by $1/|M|$. The $|M|$ powers cancel, and we're left with a blockwise adjustment of $m_i / n_i$. We again use the blockwise structure to calculate the adjoint. First, define $\Theta : S_1^{n_i m_i} \rightarrow S_1^{m_i m_i}$ by $\Theta_i(\rho) \equiv \id_{m_i} / m_i \otimes \tr_{n_i}(\rho)$. We calculate
\begin{equation}
\Theta_i^\dagger(\sigma) =  \id_{n_i} \otimes \tr_{m_i}(\sigma),
\end{equation}
where the $\tr_{m_i}$ is positioned to remove the $\id_{m_i}/m_i$ attached by $\Theta_i$ when composing $\Theta_i^\dagger \Theta_i$. We may also directly calculate
\begin{equation}
R_{\id/n_i m_i, \Theta_i} = \frac{m_i}{n_i} \Theta_i^\dagger
\end{equation}
This yields the result that
\begin{equation}
\E_\N^{c \dagger}(\sigma) =  \oplus_i (1_{n_i} \otimes \tr_{m_i}(P_i \sigma P_i)) \pl,
\end{equation}
by which we obtain equation \eqref{eq:petzform}. Then
\begin{equation}
\E_\N^{c \dagger} \circ \E_\N^{c}(\rho) = \oplus_i (1_{n_i} / m_i \otimes \tr_{n_i}(P_i \rho P_i)) \pl.
\end{equation}
Thus we find
\begin{equation}
R_{\id/|M|, \E_\N^c} \circ \E_\N^c(\rho) = \oplus_i (1_{n_i} / n_i \otimes \tr_{n_i}(P_i \rho P_i)) = \E_{\N'}(\rho) \pl.
\end{equation}
To show that $\E_\N^c$ recovers $R_{\id/|M|, \E_\N^c}$ for default state $\E_\N^c(\id/|M|)$, we use the explicit form of the Petz map's Petz map. For convenience, let $R = R_{\id/|M|, \E_\N^c}$.
\begin{equation} \label{eq:recovrecov}
R_{\E_\N^c(\id/|M|), R} = 
	\Big ( \E_\N^c \Big ( \frac{\id}{|M|} \Big ) \Big )^{1/2} R^{\dagger} \Big ( \Big ( \frac{\id}{|M|} \Big )^{-1/2}
	 \sigma \Big ( \frac{\id}{|M|} \Big )^{-1/2} \Big ) \Big ( \E_\N^c \Big ( \frac{\id}{|M|} \Big ) \Big )^{1/2}
\end{equation}
Up to normalization, $R^\dagger = \E_\N^c$. Considering normalization, the normalization constants in equation \eqref{eq:recovrecov} directly cancel those in \eqref{eq:petzformula}. Hence $R_{\E_\N^c(\id/|M|), R} = \E_\N^c$.
\end{proof}

\begin{proof} (of Corollary \ref{cor:ssa})
First, by the formula for relative entropy in finite dimensions, we expand the first line in equation \eqref{ssa1} as
\begin{equation} \label{eq:diffentexpand}
\begin{split}\tr \big (& \rho \log \rho - \E_\S(\rho) \log \E_\S(\rho) - \E_\T(\rho) \log \E_\T(\rho) + \E_{\S \cap \T}(\rho) \log \E_{\S \cap \T}(\rho) \\
 - & \rho \log \sigma + \E_\S(\rho) \log \E_\S(\sigma) + \E_\T(\rho) \log \E_\T(\sigma) - \E_{\S \cap \T}(\rho) \log \E_{\S \cap \T}(\sigma) \big ) \pl.
\end{split}
\end{equation}
The 1st line in equation \eqref{eq:diffentexpand} is equal to $H(\S)_\rho + H(\T)_\rho - H(\M)_\rho - H(\S \cap \T)_\rho$. Since $\E_\T(\sigma) = \sigma$ and by the defining property of conditional expectations,
\begin{equation}
\tr(\E_\T(\rho) \log \E_\T(\sigma)) = \tr(\rho \E_\T(\log \E_\T(\sigma))) = \tr(\rho \log \sigma) \pl .
\end{equation}
The last equality follows from Lemma \ref{lem:condexplog}. Similarly, and using the commuting square, 
\begin{equation}
\tr(\E_{\S \cap \T}(\rho) \log \E_{\S \cap \T}(\sigma)) = \tr(\E_\T \E_\S(\rho) \log \E_\T \E_\S(\sigma)) = \tr(\E_\S(\rho) \log \E_\S(\sigma)) .
\end{equation}
Hence the positive and negative terms in the 2nd line of equation \eqref{eq:diffentexpand} cancel. This proves the equality in equation \eqref{ssa1}.

Suppose $\E_{\S}\circ\E_{\T}=\E_{\T}\circ\E_{\S}=\E_{\S \cap \T}$. We have
\begin{equation} \label{eq:ssapf}
\begin{split}
  &H(\S)_\rho- H(\M)_\rho + H(\T)_\rho - H(\S \cap \T)_\rho 
\\=&D^{\S}(\rho)-D^{\S}(\E_{\T}(\rho))\ge 0 \pl.
\end{split}
\end{equation}
The last step follows from $\E_{\T}(\S)= (\S \cap \T) \subset \S$ and the well-known data processing inequality for relative entropy, and Lemma \ref{lem:entropyforms}. This proves that the commuting square implies the inequality. To get the other direction, for any state $\si\in \S$, the inequality \eqref{ssa1} implies that
\begin{align*}\D^{\S \cap \T}(\E^{\T}(\si))=H(\E_{\S \cap \T}(\si))-H(\E_{\T}(\si))=0\pl,\end{align*}
which implies $\E_{\T}(\si)=\E_{\R}(\si)$. Thus for arbitrary $\rho\in \M$,
\begin{align*}\E_{\T}\circ\E_{\S}(\rho)=\E_{\S \cap \T}\circ\E_{\S}(\rho)=\E_{\S \cap \T}(\rho)\pl.\end{align*}
That $\E_{\S}(\rho)\circ\E_{\T} = \E_{\S \cap \T}(\rho)$ follows similarly.
\end{proof}

\subsection{Properties of $D^\N$} \label{sec:dnprops}
In this section we briefly show that $D^\N$ is an appropriate monotone for a resource theory.
\begin{prop} \label{prop}$D^\N$ satisfies the following properties and is thereby a resource measure.
\begin{enumerate}
\item[i)]{\rm(Non-negativity)}: $D_\alpha^\N(\rho)\ge 0$ for any state $\rho\in \M$ and $1/2 \leq \alpha \leq \infty$.
\item[ii)]{\rm(Faithfulness)}:  $D^\N(\rho)=0$ if and only if $\rho \in \N$.
\item[iii)]{\rm(Monotonicity)}: Let $\Phi:\M\to\M$ be a completely positive trace preserving map. If $\Phi(\N)\subset \N$, then
\[D_\alpha^\N(\Phi(\rho))\leq D_\alpha^\N(\rho)\pl.\]
\end{enumerate}
\end{prop}
\begin{proof}
Properties i) and ii) follow easily from the properties of sandwiched relative R\'enyi entropy \cite{muller-lennert_quantum_2013,wilde_strong_2014}. For iii),
\begin{align*}
D_\alpha^\N(\rho)=\inf_\si D_\alpha(\rho||\si) \geq \inf_\si D_\alpha(\Phi(\rho)||\Phi(\si)) \geq \inf_\si D_\alpha(\Phi(\rho)||\si) = D_\alpha^\N(\Phi(\rho))\pl.
\end{align*}
The first inequality above is the data processing inequality and the second inequality follows from the fact $\Phi$ sends states of $\N$ to states of $\N$.
\end{proof}
In particular, $D^\N$ is closely connected to known measures of asymmetry. To illustrate this, we may consider $\N$ to be the algebra generated by observables that are invariant under the action of a locally compact group $G$ under some particular representation. In this case,
\begin{equation}
\E_\N(\rho) = \int_G U_g \rho U^\dagger_g d \mu(g),
\end{equation}
in which $\mu(g)$ is the $G$-invariant Haar measure over $G$, and $U_g$ is the unitary representation of $g \in G$. Then $D^\N(\rho) = H(\E_\N(\rho)) - H(\rho)$ is the Holevo asymmetry measure defined in \cite{gour_measuring_2009,vaccaro_tradeoff_2008}. There are however some subtle differences between this quantity and the notion of ``frameness" usually considered. Primarily, we would generally assume that for $n$ copies of the system, we would take the algebra $\N^{\otimes n} \subset \M^{\otimes n}$. For frameness, however, it is common to assume that the same reference frame governs all copies of the physical system \cite{gour_measuring_2009}, which is not consistent with taking $n$ copies of the invariant subalgebra.

\subsection{Square calculus and Information Quantities}
In this section, we will assume that capital letters $A, B, C, ...$ may denote subsystems or algebras. For $x=\bigg[\begin{array}{cc}A\!&M\\ C\!& B\end{array}\bigg]$ and $y=\bigg[\begin{array}{cc}\tilde{A}\!& \tilde{M}\\ \tilde{C}\!& \tilde{B}\end{array}\bigg]$, we say that $(x\ten y,\si)$ is an \emph{extension} of $(x,\rho)$  if
$\si$ is a density for $M\tilde{M}=M\ten \tilde{M}$ and $id\ten \tau_{\tilde{M}}=\rho$. We may then define
 \[ I_{ext}
 \bigg[\begin{array}{cc}A\!&M\\ B\!& C\end{array}\bigg]_{\rho}
 \lel \inf_{(\id
 \ten tr_{\tilde{M}})\si=\rho} I(x\ten y)_{\si} \pl, \]
where the infimum is taken over all commuting squares $y$.

\begin{prop} \label{prop:extmonogamy}
Let $x_1 = \bigg[\begin{array}{cc}A\!&M\\ C\!& B\end{array}\bigg]$, and $x_2 = \bigg[\begin{array}{cc}A\!&M\\ C\!& B\end{array}\bigg]$. Then
  \[ I_{ext}(x_1\ten x_2)\gl I_{ext}(x_1)+I_{ext}(x_2) \pl .\]
\end{prop}

\begin{proof} Let $z=\bigg[\begin{array}{cc}S\!&N\\ R\!& T\end{array}\bigg]$ be  commuting square. Thanks to Lemma \ref{mo1}, we have
\begin{align*}
I\bigg[\begin{array}{cc}A\tilde{A}S\!&M\tilde{M}N\\ C\tilde{C}R\!& B\tilde{B}T\end{array}\bigg]
 &=I\bigg[\begin{array}{cc}A\tilde{A}S\!&M\tilde{M}N\\ C\tilde{A}R\!& B\tilde{M}T\end{array}\bigg]
 +I\bigg[\begin{array}{cc}C\tilde{A}R\!&B\tilde{M}T\\ C\tilde{C}R\!& B\tilde{B}T\end{array}\bigg]
\end{align*}
Note that the tensor product of commuting squares is a commuting square and hence
 \[ \bigg[\begin{array}{cc}\tilde{A}S\!&\tilde{M}N\\ \tilde{A}R\!& \tilde{M}T\end{array}\bigg]
 \lel \bigg[\begin{array}{cc}\tilde{A}\!&\tilde{M}\\ \tilde{A} \!& \tilde{M}\end{array}\bigg] \ten
 \bigg[\begin{array}{cc}S\!&N\\ R\!& T\end{array}\bigg] \]
and
 \[ \bigg[\begin{array}{cc}CR\!&BT\\ CR\!& BT\end{array}\bigg]
 \lel
 \bigg[\begin{array}{cc}C\!&B\\ C\!& B\end{array}\bigg]\ten \bigg[\begin{array}{cc}R\!&T\\ R\!& T\end{array}\bigg] \]
are also commuting squares. The infimum in $I_{ext}$ completes the proof.
\end{proof}
\begin{rem} In some situations we may take the infimum over a smaller class of commuting squares. The proof above works as long as the trivial objects $\bigg[\begin{array}{cc}C\!&B\\ C\!& B\end{array}\bigg]$ are still in this class. Similarly, we may extend our definition using more liberal definitions for $AB$, for example by allowing free and tensor products. As long as this definition is symmetric the argument above remains valid.
\end{rem}
Let us introduce obviously free operations under which mutual information is non-increasing. Our starting point is
$x=(\bigg[\begin{array}{cc}A\!&M\\ C\!& B\end{array}\bigg],\rho)$, $\rho$ a state $\rho$ on $M$ and $I(x)=I\bigg[\begin{array}{cc}A\!&M\\ C\!& B\end{array}\bigg]_{\rho}$.

\begin{enumerate}
\item[E)] ({\bf Tensor extension by commuting squares}). Let $y = \bigg[\begin{array}{cc}S\!&N\\ R\!& S\end{array}\bigg]$ be a commuting square and $\si=\rho\ten \tilde{\rho}$ be a tensor product state. Then $I(x)=I(x\ten y)$.
\item[A)] ({\bf Automorphism}) Let $\al:M\to \tilde{M}$ be trace preserving automorphism such that $\al(A)=\tilde{A}$, $\al(B)=\tilde{B}$ and $\al(C)=\tilde{C}$. Then $I(x)=I(\tilde{x})$ holds for $\tilde{x}=(\tilde{C},\tilde{A},\tilde{B},\tilde{C},\tilde{M},\al(\rho))$.
\item[Tr)] ({\bf Special partial traces}). In case $x=(C\ten 1,A\ten D,B\ten 1,M\ten D,\rho)$, we have
         \[ I(x)\gl I(y)\]
for $y=(C,A,B,M,id\ten tr_G(\rho))$. Similarly, we can interchange the roles of $A$ and $B$.
\end{enumerate}

For the last operation, we observe that
\[ I\bigg[\begin{array}{cc}A\ten G\!&M\ten G\\ C\ten 1 \!&B\ten 1 \end{array}\bigg]_{\rho}
 \lel D(\rho|\E_{A\ten G}(\rho))+ H(B\ten 1)_{\rho}-H(C\ten 1)_{\rho} \pl .\]
By the data processing inequality we have
 \[ D(\id \ten tr_G(\rho)|\id \ten tr_G(\rho))
 \kl  D(\rho|\E_{A\ten G}(\rho)) \pl .\]
The equality
 \[ H(B\ten \id)_{\rho}-H(C\ten \id)_{\rho}\lel
 H(B)_{\id \ten tr_G(\rho)}-H(C)_{\id \ten \rho(\rho)} \]
is easily checked be canceling $\ln tr(1_G)$.

\begin{rem} \label{rem:secretshare}
{\rm In case $x=(C,A\ten D,B\ten \id,M\ten D)$ we can make $D$ public without increasing $I$. Indeed, we have
\begin{align*}
  I\bigg[\begin{array}{cc}A\ten D\!&M\ten N\\ C\ten 1 \!&B\ten 1 \end{array}\bigg]
 &= I\bigg[\begin{array}{cc}A\ten D\!&M\ten N\\ C\ten D \!&B\ten D \end{array}\bigg]
 +I\bigg[\begin{array}{cc}C\ten D\!&B\ten D\\ C\ten 1 \!&B\ten 1 \end{array}\bigg] \\
 &\gl I\bigg[\begin{array}{cc}A\ten D\!&M\ten N\\ C\ten D \!&B\ten D \end{array}\bigg]\pl, \end{align*}
because the right hand square is commuting.
}\end{rem}

We have seen above that the alternating operations given by trivial extensions, automorphism and special partial traces decrease $I$. For averages of more general partial traces, we find monotonicity for $I^{sq}$:

\begin{prop} {\rm \bf  (Traces with respect to commuting squares)} Let $x=\bigg[\begin{array}{cc} A \!& M\\ C \!& B   \end{array}\bigg]_{\phi(\rho)}$ and $y=
 \bigg[\begin{array}{cc} S \!& N\\ R \!& T   \end{array}\bigg]_{\rho}$ such that $y$ is a commuting square. Then
 \[ I^{sq}(x)_{id\ten tr_N\rho}
  \kl I^{sq}(x\ten y)_{\rho} \pl .\]
\end{prop}

\begin{proof} We just note that if $y$ and $z$ are commuting squares, then so is $(x\ten z)$, and hence
 \[ I^{sq}(x\ten y) \lel \inf_{z} I(x\ten y\ten z)
 \lel \inf_z I(x\ten (y\ten z))
 \gl \inf_{\tilde{z}}I(x\ten \tilde{z}) \lel I^{sq}(x) \pl .\]
This implies the assertion.
\end{proof}

\begin{theorem} $I^{sq}$ is non-increasing under trivial extensions, automorphism and partial traces with respect to commuting squares.
\end{theorem}

\subsubsection{Module maps} Thanks to Stinespring's factorization theorem completely positive maps in finite dimension can be decomposed into three building blocs: trivial extensions, automorphism by unitaries and partial traces. We aim to achieve a similar result for module maps.
\begin{prop} \label{prop:bimod} Let $S,T\subset \mz_n$ be subalgebras and $\phi:\mz_n\to \mz_n$ be a completely positive unital map such that
 \[ \phi(axb) \lel a\phi(x)b \]
holds for $a,b\in T$ and $\phi(S)\subset S$. Then there exists a finite dimensional Hilbert space $D$, a unit vector $d\in D$  and a unitary $U=U_1U_2$, a $^*$ representation $\pi:T\to \mz_n\ten \mathbb{B}(D)$  such that
 \begin{enumerate}
 \item[i)] $U_1\in S'\ten \mathbb{B}(D)$;
  \item[ii)] $U_2\in S\ten \mathbb{B}(D)$;
  \item[iii)] $(t\ten 1)U_1=U_1\pi(t)$ for all $t\in T$;
  \item[iv)] $\pi(t)U_2\lel U_2(t\ten 1)$ for all $t\in T$;
   \item[iv)] $(t\ten 1)U=U(t\ten 1)$;
   \item[v)] $id\ten \om_d(U^*(x\ten 1)U)=\phi(x)$ holds for $x\in \mz_n$.
 \end{enumerate}
Conversely, the properties define a ucp map which is $T$-bimodule and satisfies $\phi(S)\subset S$.
\end{prop}

\begin{proof} As usual we consider the inner product
 \[ (x\ten h,y\ten k) \lel (h,\phi(x^*y)k) \pl \]
and $N=\{\xi|(\xi,\xi)=0\}$. We denote by $[\xi]$ the equivalence classes of $\xi$ in $H=\mz_n\ten \ell_2^n/N$ and $x\ten_{\phi}h=[x\ten h]$. We note that $ya\ten_{\phi} h=y\ten_{\phi} ah$ for $a\in T$. Let $\pi:\mz_n\to \mathbb{B}(H)$ be the representation $\pi(z)(x\ten h)=zx\ten h$ which remains well-defined and unital on $H$. Therefore $H\lel \ell_2^n\ten D$ for some Hilbert space $D$ and $\pi(\mz_n)'=\mathbb{B}(D)$. Now we may consider the copy $K=\{ [1\ten h]| h\in \ell_2^n\} \subset H$. Note that the Stinespring isometry is given by the
inclusion map $V:\ell_2^n\to H$, $V(h)\lel 1\ten_{\phi}h$ and that
 \[ V(ah) \lel 1\ten_{\phi}ah \lel a\ten_{\phi} h \lel (a\ten 1) V(h) \pl .\]
This means $V$ is $T$ module map. Let $d\in D$ be a unit vector and $P=id\ten P_d:\ell_2^n\ten D\to \ell_2^n\ten D$ be obtained from the orthogonal projection onto $d$. Then $\tilde{V}=VP_d \in \mathbb{B}(\ell_2^n\ten D,\ell_2^n\ten D)$ is a partial isometry in $(T\ten 1)'$ and hence there exists a unitary $U \in (T\ten 1)'$ extending it. This means for a $T$-module map $\phi$ we have found the dilation
  \[ \phi(x) \lel w_d(U^*(x\ten 1_D)U) \]
given by $w_d(a\ten b)=a(d,bd)$ satisfying with $U\in (T\ten 1)'$. However, to accommodate the second condition we have to choose $U$ more carefully. First we consider $\phi:S\to S$, Since $S=\sum_k M_{n_k}$, we can understand $\phi$ as a family of ucp maps, and find a $v\in C_m(S)$ such that
 \[ v^*(x\ten 1_d)v\lel \phi(x) \pl .\]
We may assume that $m\le {\rm dim} D$ used above, and in fact, we may then assume $d={\rm dim}(D)$, by enlarging it. Then we note that
  \[ \{\sum_k x_k\ten_{\phi} h_k| x_k\in S, h_k\in \ell_2^n\} \lel
   X_{\phi}\cong \pi(S)v(\ell_2^n)\subset \ell_2^n\ten D \pl.\]
Observe that the orthogonal projection $P:\ell_2^n\ten D\to \pi(S)v(\ell_2)$ is a projection onto the $\pi(S)$ module and hence $P=v^*v$ is in $S'\ten \mathbb{B}(D)\ten \cap S\ten \mathbb{B}(D)=(S\cap S')\ten \mathbb{B}(D)$. Here $S\cap S'$ is the center of $S$. Let us now  fix a unit vector $d$ consider the map $\si:\pi(S)v(\ell_2^n\ten d)\to \pi(S)V(\ell_2^n\ten d))$, $V$ constructed above defined by
 \[ \si(\pi(s)v(h\ten d)))\lel \pi(s)V(h\ten d) \pl .\]
Since the space $S\ten_{\phi}\ell_2^n$ naturally embeds into $\mz_n\ten_{\phi}\ell_2^n$, we deduce that $\si$ is an isometry between $\pi(S)=S\ten 1$ modules and hence $\si P$ in $S'\ten \mathbb{B}(D)$. Therefore we find a unitary $U_1\in S'\ten \BB(D)$ extending this partial isometry. Let us define $\hat{\pi}(t)\lel U_1^*(s\ten 1)U_1$ and
  \[ Y \lel \{y:\ell_2^n\ten D\to \ell_2^n\ten D| \forall_{t\in T}:(t\ten 1)y\lel y\hat{\pi}(t)\} \pl .\]
Then $Y$ is a TRO, i.e. $y_1y_2^*y_3\in Y$ for all $y_1,y_2,y_3\in Y$. Moreover, the intersection
 \[ X\lel S\ten \BB(D) \cap Y \]
is also a TRO, see \cite{gao_capacity_2018,eleftherakis_stable_2008} for more information. Then we note that
 \begin{align*}
  U_1 v(th\ten d) &=
    V(th) \lel(t\ten 1)V(h\ten d) \lel (t\ten 1)\si v(h\ten d) \\
  &= (t\ten 1) U_1 v(h\ten d) \lel U_1 \hat{\pi}(t)v(h\ten d) \pl .
  \end{align*}
Therefore $v\in X$ and $e=v^*v\in X^*X$. Since $X$ is a finite dimensional TRO, $X^*X$ is a unital $C^*$ algebra. Therefore we can find partial isometries $\gamma_j$ with disjoint range such that
 \[ \sum_j \gamma_j^*e\gamma_j \lel 1_{X^*X} \pl .\]
In other words $U_2\lel \sum_j v\gamma_j\in X$ is a unitary. This means we find a unitary  \[ U\lel U_1U_2 \]
such that $U_1\in S'\ten \mathbb{B}(D)$, $U_2\in S\ten \BB(D)$, and
 \[ (t\ten 1)U\lel (t\ten 1)U_1U_2\lel
 U_1\hat{\pi(t)}U_2\lel U_1U_2(t\ten 1) \pl .\]
This means this new $U$ is a different of $U\in T'\ten \mathbb{B}(D)$ such that
 \[ U(\ell_2^n\ten d) \lel V(\ell_2^n\ten d) \pl .\]
This implies
  \[\phi(x)\ten |d\ran\lan d| \lel    V^*(x\ten 1)V\lel  (1\ten P_d) U^*(x\ten 1)U(1\ten P_d) \pl .\]
We leave it to the reader to verify that the conditions above characterize $T$-module maps with $\phi(S)\subset S$.
\end{proof}

\begin{cor} \label{cor:bimodtechnical}Let $\phi:S_1^n\to S_1^n$ be a completely positive trace preserving map such that $\phi^{\sharp}$ is unital $T$-module map with $\phi^{\sharp}(S)\subset S$, $\phi^{\sharp}(M)\subset M$ and $R\subset T$, $S,T\subset M$. Then
\[
 I\bigg[\begin{array}{cc} S \!& M\\ R \!& T   \end{array}\bigg]_{\phi(\rho)}
 \kl I\bigg[\begin{array}{cc} S \!& M\\ R \!& T   \end{array}\bigg]_{\rho} \pl .\]
Similarly, $I^{sq}$ is non-increasing under $T$-bimodule maps which preserve $S$.
\end{cor}

\begin{proof} We may replace $\mz_n$ in the proof above by $M\subset \mz_n$ and assume that $V\in M\ten\mathbb{B}(D)$. We find $U=U_1U_2$ such that $U_2\in S\ten \mathbb{B}(D)$ and $U_1\in M\cap S'\ten \mathbb{B}(D)$. Then we proceed in three  steps. We first add the state $\rho\mapsto U_2(\rho\ten |d\ran\lan d|)U_2^*$, $T\mapsto U_2(T\ten 1)U_2^*$, $R\mapsto U_2(R\ten 1)U_2^*$, $S\mapsto S\ten \mathbb{B}(D)$, $M\mapsto M\ten \mathbb{B}(D)$. This leaves $I$ invariant. In the next step we apply the unitary $U_1$ and see that $R$ and $T$ are back in the $R\ten 1$, $T\ten 1$ position and hence
 \[
 I\bigg[\begin{array}{cc} S\ten \mathbb{B}(D) \!& M\ten \mathbb{B}(D)\\ R \ten 1 \!& T \ten 1  \end{array}\bigg]_{U\rho\ten |d\ran\lan d|U^*}
 \lel I\bigg[\begin{array}{cc} S \!& M\\ R \!& T   \end{array}\bigg]_{\rho} \pl .\]
Now, we apply the special partial trace over $D$ and deduce the assertion, thanks to data processing. The additional assertion follows by applying the first statement to $x\ten y$ and $\phi\ten id$.
\end{proof}

\subsection{Individual Operations} \label{sec:opproofs}
\begin{proof} (Theorem \ref{B}) We first consider pure $\rho^{ABC}$. Let $\E_\S(\cdot)=id_A\ten tr_F(V^* \cdot V)$ be a Stinespring dilation of $\E_\S$ such that $F$ is the environment system and $\E_\S^c(\cdot) = tr_A\ten id_F(V^* \cdot V)$ is the complementary channel. Denote $\tilde{\rho}^{FABC}=(V\ten 1_{BC})\rho^{ABC}(V\ten 1_{BC})^*$. Since $\tilde{\rho}$ is also pure, we have
 \begin{align*}H(F|B)_{\tilde{\rho}}=& -H(F|AC)_{\tilde{\rho}}
 \\ =& -H(FAC)_{\tilde{\rho}}+H(AC)_{\tilde{\rho}}
 \\ =& -H(AC)_\rho+H(\S C)_\rho
 \\ \ge & -H(\T C)_\rho+H(\R C)_\rho\pl.
 \end{align*}
The third inequality is because $\tilde{\rho}^{FAC}=(V\ten 1_C)\rho^{AC}(V\ten 1_C)^*$ and \[\tilde{\rho}^{AC}=tr_F((V\ten 1_C)\rho^{AC}(V\ten 1_C)^*)=\E_\S\ten id_C (\rho^{AC})\pl.\] The last inequality is the SSA inequality for the commuting square $\scriptsize\kla\begin{array}{ccc} \S C& \subset &AC\\ \cup   & & \cup \\ \R C & \subset &  \T C \end{array}\mer$\pl. This inequality
\[H(F|B)_{\tilde{\rho}} \ge  -H(\T C)_\rho+H(\R C)_\rho\]
extends to all states $\rho$ because the LHS is concave and RHS is convex of $\rho$. Rewriting it we obtained \eqref{tripa}.
\end{proof}

\begin{lemma} \label{lem:envcondexp}
Let $\S, \T \subset \M$ be subalgebras such that $[\E_{\S'}, \E_\T] = 0$. Then $\E^c_{\T} \E_{\S} \R_{\id/d, \E^c_{\T}}$ is idempotent and is its own Petz map.
\end{lemma}
\begin{proof}
For idempotence, we apply Lemma \ref{lem:petzrecov} and calculate,
\begin{equation}
\begin{split}
\E_{\S}^c \E_{\T} R_{\id/d, \E_\S^c}  \E_{\S}^c \E_{\T} R_{\id/d, \E_\S^c} 
	= \E_{\S}^c \E_{\T} \E_{\S'} \E_{\T} R_{\id/d, \E_\S^c}
	= \E_{\S}^c \E_{\T} R_{\id/d, \E_\S^c} \pl.
\end{split}
\end{equation}
To show that this is its own Petz map, we use the decomposition of Petz maps for channels composed in series. $\E_{\T}$ is its own Petz map. $\E_{\S}^c$ has Petz map $R_{\id/d, \E_\S^c}$ by definition, and by Lemma \ref{lem:petzrecov}, $R_{\id/d, \E_\S^c}$ has Petz map $\E_{\S}^c$.
\end{proof}

\begin{lemma} \label{lem:envcompose}
Let $\Phi : S_1(\M_0) \rightarrow S_1(\M_1)$, and $\Psi : S_1(\M_1) \rightarrow S_1(\M_2)$ be quantum channels. Let $E_\Phi$ and $E_\Psi$ be their respective environment systems (not denoting conditional expectations or other channels) in minimal Stinespring dilations, and $E_{\Psi \circ \Phi}$ be the minimal environment of $\Psi \circ \Phi$ Then there is an isometry $V : E_{\Psi \circ \Phi} \rightarrow E_\Phi \otimes \tilde{E}_\Psi$, for which $\tilde{E}_\Psi \cong E_\Psi$, and $E_\Phi$ contains the output of $\Phi^c$.
\end{lemma}
\begin{proof}
Let $U_\Phi$ and $U_\Psi$ be the respective Stinespring isometries of $\Phi$ and $\Psi$. Under $U_\Psi \circ U_\Phi$,
\begin{equation}
S_1(\M_0) \rightarrow_{U_\Phi} S_1(\M_1) \otimes S_1(E_\Phi) \rightarrow_{U_\Psi \otimes \id_{E_\Phi}} S_1(\M_2) \otimes S_1(E_\Phi) \otimes S_1(E_\Psi) \pl.
\end{equation}
Since $U_\Psi$ has no effect on $E_\Phi$, that system still contains the environment of $E_\Phi$. Similarly, $S_1(\M_2)$ contains the output of $\Psi \circ \Phi$. As shown in \cite{holevo_complementary_2007}, complementary channels are unique up to partial isometry, and as noted in \cite{crann_private_2016}, there is an isometry mapping a minimal environment to any other environment.
\end{proof}

\begin{proof} (of Theorem \ref{thm:dual})
First,
\begin{equation}
\begin{split}
I(\S : \T \subset \S \T)_\rho
	& =  H(\S) + H(\T) - H(\S \cap \T) - H(\S \T)
\\	& = D(\E_{\S \T}(\rho) \| \E_\S(\rho)) - D(\E_\T(\rho) \| \E_{\S \cap \T}(\rho))
\\	& = D(\E^c_{(\S \T)'}(\rho) \| \E^c_{(\S \T)'} \E_\S(\rho)) - D(\E^c_{\T'}(\rho) \| \E^c_{\T'} \E_{\S \cap \T}(\rho)) .
\end{split}
\end{equation}
In the first term of the last line, we use that $\R_{\id/d, \E^c_{(\S \T)'}} \E^c_{(\S \T)'} \E_\S = \E_\S$, and $\R_{\id/d, \E^c_{(\S \T)'}} \E^c_{(\S \T)'} = \E_{\S \T}$ by Lemma \ref{lem:petzrecov}, so the application of $\E^c_{(\S \T)'}$ to both arguments of $D(\cdot \| \cdot)$ is reversible by its Petz recovery. A similar argument holds for $\E^c_{\T'}$ in the second term. Application of a fully recoverable channel leaves relative entropy invariant by data processing in both directions. Let $\tilde{\E}_\S = \E^c_{(\S \T)'} \E_\S \R_{\id/d, \E^c_{(\S \T)'}}$, and $\tilde{\E}_{\S \cap \T} = \E^c_{\T'} \E_{\S \cap \T} \R_{\id/d, \E^c_{\T'}}$. By Lemma \ref{lem:envcondexp}, these are idempotent and self-recovering. We then have
\begin{equation}
\begin{split}
I(\S : \T \subset \S \T)
	& = D(\E^c_{(\S \T)'}(\rho) \| \tilde{\E}_\S \E^c_{(\S \T)'} (\rho)) - D(\E^c_{\T'}(\rho) \| \tilde{\E}_{\S \cap \T} \E^c_{\T'} (\rho))
\\	& = H(\tilde{\E}_\S \E^c_{(\S \T)'}(\rho)) - H(\E^c_{(\S \T)'}(\rho)) + H(\E^c_{\T'}(\rho)) - H(\tilde{\E}_{\S \cap \T} \E^c_{\T'} (\rho)) 
\\	& = H(\E^c_{(\S \T)'} \E_\S (\rho)) - H(\E^c_{(\S \T)'}(\rho)) + H(\E^c_{\T'}(\rho)) - H(\E^c_{\T'} \E_{\S \cap \T} (\rho)) \pl.
\end{split}
\end{equation}
Since $\ketbra{\psi}$ is pure, the middle two terms equal $H(\T')_{\ketbra{\psi}} - H(\S' \cap \T')_{\ketbra{\psi}}$. We are done with these terms.

We turn our attention to the outer terms. First, $\E^c_{\T'} \E_{\S \cap \T} = \E^c_{\T'} \E_{\T} \E_{\S} = \E^c_{\T'} \E_{\S}$. Thereby, these terms become
\begin{equation}
H(\E^c_{(\S \T)'} \E_\S (\rho)) - H(\E^c_{\T'} \E_{\S} (\rho))
\end{equation}
Since there exists an isometry from the complementary channel of the minimal Stinespring dilation to any other, we are free to Stinespring dilate non-minimally without changing these entropies. Here we will use $E_\S$ to denote the environment system of $E_\S$ (not a channel), while $\E_\S$ denotes the conditional expectation. Let $E_{(\S \T)'}$ be the environment of $\E_{(\S \T)'}$. Since $\E^c_{(\S \T)'} = \E^c_{\S' \cap \T'} = (\E_{\T'} \E_{\S'})^c$, we may compose the environments as $E_{(\S \T)'} \cong E_{\T'} \otimes \tilde{E}_{\S'}$ (which is a tensor system, not a channel composition) by Lemma \ref{lem:envcompose}. We then use the fact that $\E^c_{\T'} = (\E_{\T'} \E_{\T' \S'})^c$ to rewrite $E_{\T'} = E_{\T' \S'} \otimes \tilde{E}_{\T'}$, as though the channel $\E_{\T' \S'}$ had been applied first, and the channel $\E_{\T'}$ to the output of that, keeping both environments. In the term containing $\E^c_{\T'}$, we expand the environment to $E_{\T' \S'} \otimes \tilde{E}_{\T'}$. We also further expand $E_{(\S \T)'} \cong E_{\T'} \otimes \tilde{E}_{\S'} \cong E_{\T' \S'} \otimes \tilde{E}_{\T'} \otimes \tilde{E}_{\S'}$. This leaves us with
\begin{equation}
\begin{split}
& H(E_{\T' \S'} \otimes \tilde{E}_{\T'} \otimes \tilde{E}_{\S'})_{\E_\S(\rho)} - H(E_{\T' \S'} \otimes \tilde{E}_{\T'})_{\E_\S(\rho)} \\
& = - D(E_{\T' \S'} \otimes \tilde{E}_{\T'} \otimes \tilde{E}_{\S'} \| E_{\T' \S'} \otimes \tilde{E}_{\T'} \otimes 1)_{\E_\S(\rho)} \\
& \leq - D(E_{\T' \S'} \otimes \tilde{E}_{\S'} \| E_{\T' \S'} \otimes 1)_{\E_\S(\rho)} \\
& = H(E_{\T' \S'} \otimes \tilde{E}_{\S'})_{\E_\S(\rho)} - H(E_{\T' \S'})_{\E_\S(\rho)}
\end{split}
\end{equation}
by data processing under partial trace. $\tilde{E}_{\T'}$ was the same system in both terms, and $\tilde{E}_{\S'}$ was split off before it. This can be written as
\begin{equation}
H((\E_{\S'} \E_{\T' \S'})^c \E_\S (\rho)) - H(\E^c_{\T' \S'} \E_\S(\rho)) \pl.
\end{equation}
We have however that $\E_{\S'} \E_{\T' \S'} = \E_{\S'}$, that $\E_{\S'}^c = \E_{\S'}^c \E_\S$, and that $\E^c_{\T' \S'} \E_\S = \E^c_{\T' \S'} \E_{\S \cap \T} \E_\S = \E^c_{\T' \S'}$, so we are left with
\begin{equation}
H(\E_{\S'}^c (\rho)) - H(\E^c_{\T' \S'}(\rho)) \pl.
\end{equation}
Since $\ketbra{\psi}$ is pure, this becomes $H(\S')_{\ketbra{\psi}} - H(\T' \S')_{\ketbra{\psi}}$. This implies $I(\S : \T \subset \S \T)_{\rho} \leq I(\S' : \T' \subset \S' \T')_{\ketbra{\psi}}$. $\M$ being a factor implies the Theorem via double-commutants.
\end{proof}

\begin{definition} (\textbf{Algebra-Modifying Individual Operations}) \label{def:algebraicops}
Let $\S, \T \subseteq \S \T$ form a commuting square. We define an individual $\S$-operation as the following steps:
\begin{enumerate}
	\item We extend $\S \rightarrow \S \otimes \C$, while $\T \rightarrow \T \otimes \CC 1$, and $\rho \rightarrow \rho \otimes \ket{0}^C$ for a fixed pure state $\ket{0}^C$.
	\item For a unitary $U : S T \otimes C \rightarrow S T \otimes C$, we may either:
	\begin{enumerate}
		\item Require that $U$ is an $\s'$-bimodule. Apply $U$ to transform $\s$ by $a \rightarrow U a U^\dagger$ for each $s \in \s$, requiring that $\S \rightarrow \tilde{\S}$ such that $\tilde{\S}, \T \subseteq \tilde{\S} \T \C$ remain a commuting square. This we call a \textbf{Heisenberg picture $\S$-unitary}.
		\item Apply $U$ to transform $\rho \otimes \ket{0}^C \rightarrow U (\rho \otimes \ket{0}^C) U^\dagger$. Transform each $a \in \M \C$ as $a \rightarrow U a U^\dagger$, under which $\S$ and $\T$ may change. This we call a \textbf{global renaming}.
	\end{enumerate}
	\item We transform $\S$ and/or $\T$ as
	\begin{enumerate}
		\item $\S \rightarrow \tilde{\S}$ such that $\tilde{\S} \subseteq \S $, and $\S \cap \T = \tilde{\S} \cap \T$.
		\item $\T \rightarrow \tilde{\T}$ such that $\T \subseteq \tilde{\T}$, and $\S \T = \S \tilde{\T}$.
	\end{enumerate}
	At this point $\s \rightarrow \tilde{\s}$ such that $\tilde{\s} \subseteq \s$, and $\t$ is unchanged. We may then transform the density at this point as $\rho \rightarrow \E_{\tilde{\S} \tilde{\T}}(\rho)$, and remove any subsystems not in a minimal factor $\M$ such that $\tilde{\S} \tilde{\T} \subseteq \M$.
\end{enumerate}
We define algebra-modifying $\T$-operations analogously.
\end{definition}
\noindent While there is some redundancy between these pictures, they differ in which operations are most convenient. An algebra-modifying $\S$-op has the form of a local Stinespring dilation as in figure \ref{fig:stine}: we add an extra system, act unitarily, and then ``trace out" by removing elements from $\S$ or locking them in $\S \cap \T$. While the global renaming unitaries do formally affect the state, they do so in a way with no physical consequences whatsoever, so we still consider them Heisenberg picture operations. If the addition of elements to $\S \cap \T$ seems odd, we are often free to restrict our attention to $\S \cap \T = \CC 1$. In these scenarios, we may always remove elements of $\S$ or $\T$.
\begin{proof} (of Theorem \ref{thm:mono})
We first show monotonicity of algebra-modifying operations. For step 1, $H(\S \T)$ and $H(\S)$ are unchanged, while $H(\T)$ and $H(\S \cap \T)$ change by an equal amount $\log |D|$. Hence this leaves $I(\S : \T)_\rho$ invariant.

For step 3a, data processing with Lemma \ref{lem:entropyforms} implies $H(\S \T)_\rho - H(\S)_\rho \leq H(\tilde{\S} \T)_\rho - H(\tilde{\S})_\rho$. For 3b, $H(\S \T)_\rho - H(\T)_\rho \geq H(\S \tilde{\T})_\rho - H(\tilde{\T})_\rho$ by data processing in the reverse direction. Since the $I$ only depends on $\rho^{\S \T}$, we may remove extra subsystems, completing step 3.

For step 2,
\begin{enumerate}[label=(\alph*)]
	\item[(a)] Let $\rho = \tr_{M'}(\ketbra{\psi}^{M M'})$, where $\M$ is a minimal factor containing $\S \T$, and $\M'$ is a purifying extra system. As an $\s'$-bimodule, $U$ does not change $\E_{\S'}(\ketbra{\psi})$, or $\E_{\S' \cap \T'}(\ketbra{\psi})$. Hence by Lemma \ref{lem:petzrecov}, it also leaves $\E^c_{\S}(\ketbra{\psi}) = \E^c_{\S}(\rho)$ and $\E^c_{\S \T}(\ketbra{\psi}) = \E^c_{\S \T}(\rho)$ unchanged. $\s' = (\S \cap (\S \cap \T)')' = \S' (\S \cap \T)$.
	\item[(b)] If $U$ is an interaction-picture global unitary, then $H(\E_{U \S U^\dagger}(U \rho U^\dagger)) = H(\E_{\S}(\rho))$, and similarly for $\T, \S \T$, and $\S \cap \T$. In essence, the unitary is simply a renaming of bases.
\end{enumerate}
This completes the analysis of algebra-modifying operations.

For state-modifying operations, $\T$-preserving property implies that $H(\T)$ and $H(\S \cap \T)$ do not change. By Lemma \ref{lem:entropyforms}, data processing, and Remark \ref{rem:bimodcomm}
\begin{equation}
\begin{split}
H(\S)_\rho - H(\S \T)_\rho & = D(\E_{\S \T}(\rho) \| \E_{\S}(\rho) )\\
	& \geq D(\Phi \circ \E_{\S \T}(\rho) \| \Phi \circ \E_{\S}(\rho) ) \\
	& = D( \E_{\S \T}(\Phi(\rho)) \| \E_{\S}(\Phi(\rho)) ) \\
	& = H(\S)_{\Phi(\rho)} - H(\S \T)_{\Phi(\rho)}.
\end{split}
\end{equation}
This proves the Theorem for bimodule operations. For adjusted bimodules, we already have monotonicity under dilation as shown for step 1 of algebraic ops, while tracing out a completely mixed subsystem in $\S$ or $\T$ is easily shown to have no effect on generalized CMI by inspecting the entropy expression.
\end{proof}
We define one more transformation on $I(\S : \T)$. We do not consider this transformation to be an operation, as it has no effect on expectation of observables in $\S$ and $\T$. Rather, it encodes a change in assumptions about the environment.
\begin{lemma} (\textbf{Heisenberg-Schr\"odinger Picture Swaps}) \label{lem:swap}
Let $\S, \T \subseteq \S \T \subseteq \M$ such that $\E_\S, \E_\T,$ and $\E_{\S \T}$ commute. Let $\R \subseteq \M$ be a subalgebra such that $\R \cap \T = \T$, and $\E_\R$ commutes with $\E_\S, \E_\T, \E_{\S \T}$, and $\E_{\S \cap \T}$. Let $\tilde{\S}$ be another algebra with conditional expectation that commutes with the aforementioned, such that $\R \cap \tilde{\S} = \R \cap \S$. Then for all $\rho \in S_1(M)$, $I(\S : \T \subseteq \S \T) \geq I(\tilde{\S} : \T \subseteq \tilde{\S})_{\E_\R(\rho)}$.
\end{lemma}
\begin{proof}
First, $H(\S)_\rho - H(\S \T)_\rho \geq H(\S)_{\E_{\R}(\rho)} - H(\S \T)_{\E_{\R}(\rho)}$ by Lemma \ref{lem:entropyforms} and data processing. Since $\R \cap \T = \T$, $(\S \T) \cap \R = (\S \cap \R) \T$. Since $\S \cap \R = \tilde{\S} \cap \R$, these terms become $H(\tilde{\S})_{\E_{\R}(\rho)} - H(\tilde{\S} \T)_{\E_{\R}(\rho)}$. Also,  $\R \cap \T = \T$ implies that $H(\T)_\rho - H(\S \cap \T)_\rho = H(\T)_{\E_{\R}(\rho)} - H(\S \cap \T)_{\E_{\R}(\rho)}$.
\end{proof}
\noindent Lemma \ref{lem:swap} allows us to expand $\S$ at the cost of degrading the state. We go from a picture with a cleaner state but less access by the party controlling $\S$ to one with a noisier but more accessible state.

\subsection{Squashed Conditional Mutual Information} \label{sec:sqcproofs}
Here we prove the properties listed in section \ref{sec:measures}.
\begin{lemma} \label{lem:convex} (property \ref{entlprop1})
Let $\R \subseteq \S, \T \subset \M$ be in commuting square. $I_{sq}(\S : \T)_{\rho}$ and $I_{conv}(\S : \T)_\rho$ are convex in $\rho$.
\end{lemma}
\begin{proof}
Let $\rho = \sum_x p_x \rho_x$ for some $x \in 1...n$. The minimizing $\tilde{M}$ system in $I_{sq}$ may attach an extra classical system such that $\sum_x p_x \rho_x \otimes \ketbra{x}$, and it is free to then extend each $\rho_x \otimes \ketbra{x}$ separately. Therefore,
\begin{equation}
I_{sq}(\S : \T)_\rho \leq \sum_x p_x I_{sq}(\S : \T)_{\rho_x} \leq \max_x I_{sq}(\S : \T)_{\rho_x}.
\end{equation}
Similarly, we note that the convex decomposition optimized over in $I_{conv}$ may replace it by $\sum_x p_x I_{conv}(\S : \T)_{\rho_x}$.
\end{proof}

\begin{lemma} \label{lem:puresq} (property \ref{entlprop2})
Let $\R \subseteq \S, \T \subsetneq \M \subset \BB(H)$ be subalgebras in commuting square for some Hilbert space $H$, $\rho \in S_1(\M)$ be a density such that $\E_{\S \T}(\rho) = \ketbra{\psi}$ be pure for some $\ket{\psi} \in H$, and $\S \T$ be a factor. Then
\begin{equation}
I_{sq}(\S : \T)_{\ketbra{\psi}} = I_{conv}(\S : \T)_{\ketbra{\psi}} = \frac{1}{2} \I(\S : \T)_{\ketbra{\psi}}.
\end{equation}
\end{lemma}
\begin{proof} (of Lemma \ref{lem:puresq})
The equality is obvious for $I_{conv}$, as there is no way to rewrite a pure state as a non-trivial convex combination.

For $I_{sq}$, via Remark \ref{rem:sqfactor}, we may assume that the infimum involves a factor in tensor position with $\S \T$. Since $\E_{\S \T}(\tilde{\rho}) = \ketbra{\psi}$ is fixed, it must remain pure. Since $\tilde{\rho}$ is pure, it must have support on only one component of the center of $\S \T$. Hence we can find a diagonal matrix block $M_i$ that is isomorphic to a factor containing the support of $\E_{\S \T}(\tilde{\rho})$. Hence we can write $\tilde{\rho} = \ketbra{\psi}^{M_i} \otimes \sigma^{C}$, where $C$ is the extension. The extension then becomes irrelevant.
\end{proof}

\begin{lemma} \label{lem:maxblock}
Let $\S, \T \subset \S \T = \M$ form a commuting square with $\M = \oplus_i (\M_i \otimes \CC 1_{m_i})$. Then $\sup_\rho I_{conv}(\S : \T)_\rho$, $\sup_\rho I(\S : \T)_\rho$, and $\sup_\rho I_{sq}(\S : \T)_\rho$ are achieved by $\rho = P_i \rho P_i$ for some $i$.
\end{lemma}
\begin{proof}
within $\M$, $\rho = \oplus_i p_i ((P_i \rho_i P_i) \otimes \id_i / m_i)$ for any given $\rho$, where $p_i = \tr(P_i \rho P_i)$, so that $\rho_i$ is a normalized density. Since $\rho$ is a convex combination,
\begin{equation}
I_{*}(\S : \T)_\rho \leq \sum_i p_i I_{*}(\S : \T)_{\rho_i} \leq \max_i I_{*}(\S : \T)_{\rho_i},
\end{equation}
where $I_{*} \in \{I_{sq}, I_{conv}\}$.
\end{proof}

\begin{lemma} \label{lem:allcomm} (property \ref{entlprop3})
Let $\S, \T \subseteq \S \T$ form a commuting square. If $\S  \T$ is a finite-dimensional commutative von Neumann algebra or if $\rho$ is a convex combination of orthogonal states, then $I_{sq}(\S : \T)_{\rho} = I_{conv}(\S : \T)_{\rho} = 0$ for all $\rho$.
\end{lemma}
\begin{proof}
By Lemma \ref{lem:maxblock}, we need only prove this for pure $\rho = \ketbra{\psi}$. Let $\M$ be a minimal factor containing $\S \T$, and $\ketbra{\psi}$ is pure in $\M$ as well. Since $\S \T$ is a commutative subalgebra of $\M$, $\S \T \subseteq \S', \T'$. Applying Theorem \ref{thm:dual}, we have $I(\S' : \T')_{\ketbra{\psi}} = 0$, because $\ketbra{\psi}$ remains pure in each commutant subalgebra, and hence in their intersection and union as well, so all entropies in this expression are zero.
\end{proof}	

\begin{lemma} \label{lem:contiuoussq} (property \ref{entlprop4})
Let $\S, \T \subseteq \S \T$ be in co-commuting square with $\S \T = \M$ a factor. If $\| \rho - \eta \|_1 \leq \epsilon < 1$, then
\begin{equation}
|I_{sq}(\S : \T)_{\sigma} - I_{sq}(\S : \T)_{\eta}| \leq 12 \sqrt{\epsilon} \log |M| + 3 (1 + 2 \sqrt{\epsilon}) h \Big (\frac{1}{1 + 2 \sqrt{\epsilon}} \Big ) \pl,
\end{equation}
and
\begin{equation}
|I_{conv}(\S : \T)_{\sigma} - I_{conv}(\S : \T)_{\eta}| \leq 6 \sqrt{\epsilon} \log |M| + 3 (1 + 2 \sqrt{\epsilon}) h \Big (\frac{1}{1 + 2 \sqrt{\epsilon}} \Big ) \pl,
\end{equation}
where $h$ is the binary entropy function.
\end{lemma}
\begin{proof}
We use Remark \ref{rem:sqfactor} to assume that the extension is via a tensor factor $\C$. Let $\M$ be a minimal factor such that $\S \T \subseteq \M$. Then
\begin{equation} \label{eq:isqcondentform}
\begin{split}
I_{sq}(\S : \T)_\rho & = \inf_{\C, \sigma} \big (H(\S \C)_\sigma + H(\T \C)_\sigma - H(\S \T \C)_\sigma - H(\C)_\sigma \big ) \\
	& = \inf_{\C, \sigma} \big ( H(\M | \C)_{\E_{\S \C}(\sigma)} + H(\M | \C)_{\E_{\T \C}(\sigma)}
		- H(\M | \C)_{\E_{\S \T} \C}(\sigma) \big )
\end{split}
\end{equation}
Since $\M$ is a factor, $H(\M | \C)$ is an ordinary conditional entropy. The rest of this proof follows Winter and Christandl's original argument as in \cite{christandl_squashed_2004}. This goes by first using the relations between fidelity and trace distance with purifications $\rho^M = \tr_{M'}(\ketbra{\psi}^{M M'})$ and $\eta^M = \tr_{M'}(\ketbra{\phi}^{M M'})$, such that $\|\ketbra{\psi} - \ketbra{\phi}^{M M'}\|_1 \leq 2 \sqrt{\epsilon}$.

For any $\sigma$ in the infimum, let $\Lambda : M' \rightarrow C$ be a quantum operation such that $(\id^M \otimes \Lambda)(\ketbra{\psi}) = \sigma^{M C}$, following the arguments of Christandl and Winter in \cite{christandl_squashed_2004}. Let $\tilde{\eta} = (\id^M \otimes \Lambda)(\ketbra{\phi})$, knowing $I_{sq}(\S : \T)_{} \leq I(\S' \C : \T' \C \subseteq \S \T \C)_{\tilde{\eta}}$. Then $\| \sigma^{M C} - \tilde{\eta}^{M C} \|_1 \leq 2 \sqrt{\epsilon}$ by monotonicity of the trace distance under application of quantum operations to both densities. We may further apply arbitrary conditional expectations to both densities. We then apply the Alicki-Fannes bound as derived in \cite{alicki_continuity_2004} and refined in \cite{winter_tight_2016} to each term in equation \eqref{eq:isqcondentform}, yielding
\begin{equation}
|I(\S' \C : \T' \C)_{\sigma} - I(\S' \C : \T' \C)_{\tilde{\eta}}| \leq 12 \sqrt{\epsilon} \log |M| + 3 (1 + 2 \sqrt{\epsilon}) h \Big (\frac{1}{1 + 2 \sqrt{\epsilon}} \Big ) \pl.
\end{equation}
We then use that $I(\S' \C : \T' \C)_{\tilde{\eta}} \geq I_{sq}(\S : \T)_\eta$ and assume without loss of generality that $I_{sq}(\S : \T)_\rho \leq I_{sq}(\S : \T)_\eta$.

We note that $I_{conv}$ admits a similar extension to that of $I_{sq}$, but where $\C$ is restricted to a classical algebra. We may thereby apply the same argument, but with the slightly better bound on classically conditioned entropies.
\end{proof}
Property \ref{entlprop6} follows from Proposition \ref{prop:extmonogamy}.

\begin{lemma} \label{lem:sqfaith} (property \ref{entlprop8})
Let $\A, \B \subset \S \T$ be factors such that $\S \subseteq \A$, $\T \subseteq \B$, $\A \otimes \B$ is valid, and $I_{sq}(\S' : \T')_\rho \leq \epsilon$. Then there exists a $k$-extension $\tilde{\sigma}^{A B_1 ... B_k}$, and for each $i \in 1...k$ there exists a subalgebra $\tilde{\T}_i \subseteq \B_i$ such that $\|\E_{\S \T}(\rho) - \E_{\S \tilde{\T}_i}(\tilde{\sigma})\|_1 \leq (k-1) \sqrt{\epsilon 2 \ln 2}$. Furthermore, there exists some separable $\sigma^{AB}$ such that $\|\E_{\S \T}(\rho^{AB}) - \E_{\S \T}(\sigma) \|_1 \leq  3.1 |B| \sqrt[4]{\epsilon}$, and some separable $\eta^{AB}$ such that $\|\E_{\S \T}(\rho^{AB}) - \E_{\S \T}(\eta) \|_2 \leq 12 \sqrt{\epsilon}$.
\end{lemma}
\begin{proof}
We begin by noting that $E_{sq}(A : B)_\rho$ admits such a $k$-extension, and norm bounds \cite{brandao_faithful_2011, li_squashed_2018}. In this case, $\S \otimes \T$ is valid, and $\E_\S(\rho^A) \otimes \E_\T(\rho^B) = \E_{\S \T}(\rho)$. By remark \ref{rem:sqfactor}, we actually have that $I_{sq}(\S : \T)_\rho = E_{sq}(A : B)_{\E_\S(\rho^A) \otimes \E_{\T}(\rho^B)}$. Hence we may trivially port the faithfulness of squashed entanglement.
\end{proof}

\begin{proof} (of Theorem \ref{thm:yesquantumsq})
By convexity of $I_{sq}$, its maximum is achieved on pure states for any fixed algebras. Hence via Lemma \ref{lem:puresq},
\begin{equation} \label{eq:maxblock}
\max_{\S, \T, \rho} I_{sq}(\S : \T)_\rho = \max_{\S, \T, \ketbra{\psi}} I(\S : \T \subset \S \T)_{\ketbra{\psi}}.
\end{equation}
By \ref{lem:maxblock}, we may restrict to the largest diagonal block in $\M$, which we will denote by $\M_0$ as an algebra or $M_0$ as a space of densities.

We will momentarily consider the expression $I(\S' : \T' \subseteq \M_0)_\rho$. Any addition to the algebras $\S'$ and $\T'$ that preserves the commuting square will affect $H(\S') - H(\S' \cap \T')$, or $H(\T') - H(\S' \cap \T')$. Via Lemma \ref{lem:entropyforms} and data processing, this does not increase $I(\S' : \T' \subseteq \M_0)_\rho$. Hence $\S' = \T' = \CC 1$ achieves the maximum value with any pure state with support in $M_0$, which is $\log |M_0|$. Given a pair of mutually unbiased bases $\X,\Z$ on $\M_0$, we may transform $\S' \rightarrow \X$, $\T' \rightarrow \Z$, but take $\ketbra{\psi}$ prepared in a third mutually unbiased basis $\Y$ of $M_0$. This maintains the $\log |M_0|$ value. $\S = \S' = \X$, and $\T = \T' = \Z$ are available for the maximum in equation \eqref{eq:maxblock}. Taking into account the factor of $1/2$ in $I_{sq}$, and noting that $I_{sq} \leq I_{conv} \leq \log |M_0| / 2$, we complete the proof.
\end{proof}

\subsection{Non-Increasing Transformations for Non-classical Measures}
\begin{proof} (of Corollary \ref{cor:sqindivops})
Consider the form,
\begin{equation}
I_{*}(\S : \T)_\rho = \inf_{\sigma, C} \big ( H(\S \C)_\sigma + H(\T \C)_\sigma - H(\S \T \C)_\sigma - H(\C)_\sigma \big )
\end{equation}
where $I_{*} \in \{I_{sq}, I_{conv}\}$ with the infimum subject to corresponding constraints. For any $\sigma$ and $C$, an individual operation does not affect $\sigma^C$, so it remains an individual operation. Hence $I(\S \C : \T \C)_\sigma \geq I(\tilde{\S} : \tilde{\T})_{\tilde{\sigma}}$ under the individual operation transforming $\S \C \rightarrow \tilde{\S} \C$, $\T \C \rightarrow \tilde{\T} \C$, and $\sigma \rightarrow \tilde{\sigma}$. The latter is a candidate in the infimum for $I(\tilde{\S} : \tilde{\T})_{\tilde{\rho}}$ under the individual operation transforming $\S \rightarrow \tilde{\S}$, $\T \rightarrow \tilde{\T}$, $\rho \rightarrow \tilde{\rho}$.
\end{proof}

\begin{proof} (of Theorem \ref{thm:entlcovar})
Let $C$ be the auxiliary system in the extension, which we can assume by remark \ref{rem:sqfactor} is in tensor position with $\S \T$. We recall
\begin{equation}
I_{*}(\S : \T)_\rho = \inf_{\C, \tilde{\rho}} H(\S \otimes \C)_{\tilde{\rho}} + H(\T \otimes \C)_{\tilde{\rho}} - H(\S \T \otimes \C)_{\tilde{\rho}} - H((\S \cap \T) \otimes \C)_{\tilde{\rho}}
\end{equation}
with some restrictions depending on the particular form of $I_{*}$, including that $\tr_C(\tilde{\rho}) = \rho$. Each $U$ in the averaging for $\E_\R$ extends to $U \otimes \id^C$ on the extended state. Since $\E_\S$ extends to $\E_{\S \otimes \C} = \E_{\S} \otimes \id^C$, it still commutes with conjugation by $U \otimes \id^C$. Similarly, each $U \otimes \id^C$ commutes with $\E_{\T \otimes \C}$, $\E_{\S \T \otimes \C}$, and $\E_{(\S \cap \T) \otimes \C}$. Because $U$ commutes with these conditional expectations and is unitary, $I(\S \C : \T \C)_{U \tilde{\rho} U^\dagger} = I(\S \C : \T \C)_{\tilde{\rho}}$. For each $\tilde{\rho}$ in the optimization of $I_{*}(\S : \T)_\rho$, $I_{*}(\S : \T)_{U \rho U^\dagger}$ achieves the same value with $U \tilde{\rho} U^\dagger$, so $I_{*}(\S : \T)_{U \rho U^\dagger} \leq I(\S : \T)_\rho$. By invertibility of $U$, they are equal.

By convexity of $I_{*}$,
\begin{equation}
I_{*}(\S : \T)_{\E_\R(\rho)} = I_{*}(\S : \T)_{\int U \rho U^\dagger d\mu(U)} \leq \int I_{*}(\S : \T)_{U \rho U^\dagger} d\mu(U) = I_{*}(\S : \T)_\rho.
\end{equation}
Returning to the explicit form,
\begin{equation}
\begin{split}
& H(\S \C)_{\E_{\R \C}(\tilde{\rho})} + H(\T \C)_{\E_{\R \C}(\tilde{\rho})}
		 - H(\S \T \C)_{\E_{\R \C}(\tilde{\rho})} - H((\S \cap \T) \C)_{\E_{\R \C}(\tilde{\rho})} \\
	= & H((\S \cap \R) \C)_{\tilde{\rho}} + H((\T \cap \R) \C)_{\tilde{\rho}}
		\\ & - H((\S \T \cap \R) \C)_{\tilde{\rho}} - H((\S \cap \T \cap \R) \C)_{\tilde{\rho}} \\
	= & H((\tilde{\S} \cap \R) \C)_{\tilde{\rho}} + H((\tilde{\T} \cap \R) \C)_{\tilde{\rho}}
		\\ & - H((\tilde{\S} \tilde{\T} \cap \R) \C)_{\tilde{\rho}} - H((\tilde{\S} \cap \tilde{\T} \cap \R) \C)_{\tilde{\rho}} \\
	= & H(\tilde{\S} \C)_{\E_{\R \C}(\tilde{\rho})} + H(\tilde{\T} \C)_{\E_{\R \C}(\tilde{\rho})}
		 - H(\tilde{\S} \T \C)_{\E_{\R \C}(\tilde{\rho})} - H((\tilde{\S} \cap \T) \C)_{\E_{\R \C}(\tilde{\rho})}
\end{split}
\end{equation}
for any $\tilde{\rho}$ and $\C$. We also have by the restrictions in either form of $I_{*}$ that $\E_{\R \C}(\tilde{\rho}) = (\E_{\R} \otimes \id)(\tilde{\rho})$ is an extension of $\E_\R(\rho)$ in $\S \T \C$. Hence $I_{*}(\S : \T)_{\E_\R(\rho)}$ and $I_{*}(\tilde{\S} : \tilde{\T})_{\E_\R(\rho)}$ optimize over the same set of states, and they achieve the same values on each. Therefore, they are equal, completing the Theorem.
\end{proof}
\begin{lemma}
Let $U$ be a unitary such that $U \S U^\dagger = \S$. Then $U$-conjugation commutes with $\E_\S$.
\end{lemma}
\begin{proof}
By assumption, $U^\dagger \E_\S (a) U \in \S$. Obviously, $\E_\S(U^\dagger a U) \in \S$ as well. If $a \in \S$, then $U^\dagger \E_\S (a) U = U^\dagger a U$. Hence $\E_\S(U^\dagger a U) = \E_\S(U^\dagger \E_\S(a) U) = U^\dagger \E_\S(a) U$ by $U \S U^\dagger = \S$. Hence when $a \in \S$, we can move the $U$-conjugation inside or outside of $\E_\S$.

In general, for any $b \in \M$, $\tr(\E_\S(U^\dagger a U) b) = \tr(a U \E_\S(b) U^\dagger) = \tr(\E_\S(a) U \E_\S(b) U^\dagger) = \tr(U^\dagger \E_S(a) U \E_\S(b)) = \tr(\E_\S(U^\dagger \E_S(a) U) \E_\S(b)) = \tr(\E_\S(U^\dagger \E_S(a) U) b) = \tr(U^\dagger \E_S(a) U b)$.

Since $\tr(\E_\S(U^\dagger a U) b) = \tr(U^\dagger \E_\S(a) U b)$ for all $a, b \in \M$, we conclude that $\E_\S(U^\dagger a U) = U^\dagger \E_\S(a) U$ for all $a \in \M$.
\end{proof}

\subsubsection{Double-correlation Extraction}
\newcommand{\obs}{\mathcal{O}}
For qubit systems $A$ and $B$, let $C_{\obs_A \rightarrow V_B}$ denote the controlled gate that performs the unitary $V_B$ if the binary observable $\obs_A$ is in the ``$-1$" eigenstate. For example, $C_{\Z_A \rightarrow X_B}$ is the standard controlled not gate when $\Z$ eigenstates define the computational basis.
\begin{lemma}
Let $A$ and $B$ be a pair of quantum systems. Let $\obs_i \in \{\X_A,\Y_A,\Z_A\}$, $V_j \in \{X_B,Y_B,Z_B\}$, $n \neq m \neq i$, and $k \neq l \neq j$, where $i,j,k,l,n,m \in \{1,2,3\}$ index the Pauli matrices $X,Y,Z$. Then
\begin{equation} 
C_{\obs_i \rightarrow V_j} \implies (\obs_n \rightarrow \obs_n \V_j), (\obs_m \rightarrow \obs_m \V_j), (\V_k \rightarrow \obs_i \V_k), (\V_l \rightarrow \obs_i \V_l)
\end{equation}
and similarly for $C_{\W_A \rightarrow V_B}$, where 
\end{lemma}
\begin{proof}
Conveniently, $W$ and $\W$ are interchangeable $\forall W \in \{X,Y,Z\}$, since the Pauli matrices happen to be Hermitian and unitary. Hence we may ignore the distinction. Also conveniently, $U_i^2 = V_j^2 = 1$ $\forall i,j$. Applying the operation $U_n$ or $U_m$ prior to this controlled gate inverts the value of $\U_i$, which flips whether or not $W_B$ would have been performed. Hence interchanging $C_{\U_i \rightarrow V_j}$ with $U_n$ or $U_m$ inverts the application of $V_j$, which is equivalent to applying $V_j$. The observable $V_j$ is performed depending on the state of $\U_i$, inverting $\V_k$ and $\V_l$. Hence $\V_k \rightarrow \U_i \V_k$ and $\V_l \rightarrow \U_i \V_l$. $\U_i$ and $\V_j$ both commute with this gate and are invariant.
\end{proof}

\begin{proof} (of Corollary \ref{cor:final})
The main trick is to use Theorem \ref{thm:entlcovar} with $R = \braket{Y_A, Y_B}$. $\E_{\braket{Y_A, Y_B}}(\rho) = \E_{\braket{\A, Y_B}} \E_{\braket{Y_A, \B}}(\rho)$ by commuting squares. Furthermore, $\E_{\braket{Y_A, \B}} = \E_{\Y_A} \otimes \id^B$, and $\E_{\Y_A}(\rho) = \frac{1}{2}(\rho + Y_A \rho Y_A)$. $Y_A \S Y_A = \braket{-X_A, -Z_A Z_B} = \S$. Similar arguments hold for $Y_B$ with $\S$ and for $\T$. Hence $\R$ is generated by a convex combination of unitaries that commute with $\E_\S$ and $\E_\T$. It is easy to see that it commutes with $\E_{\S \T}$, and $\E_{\S \cap \T}$, as the former is the whole $\A \B$ algebra, and the latter is $\CC 1$. We then check that $\R \cap \braket{Z_A, Z_B} = \CC 1 = \R \cap \braket{X_A, Z_A Z_B}$. Similarly, $\R \cap \braket{X_A, X_B} = \CC 1 = \R \cap \braket{X_A X_B, Z_B}$. The latter algebras generate the same joint, and their intersection remains $\CC 1$. Hence this is a valid transformation.

Once we have $\S = X_A, Z_A Z_B$, with $\T = X_A X_B, Z_B$, we apply the unitary $C_{Z_B \rightarrow X_A}$. This changes the algebras to $\S = \A$, and $\T = \B$. For the state,
\begin{equation}
\begin{split}
\ket{\uparrow_Y \uparrow_Y} = &  \frac{1}{2} ((\ket{00} - \ket{11}) + i (\ket{01} + \ket{10})) \\
\rightarrow  & \frac{1}{2} ((\ket{00} - \ket{01}) + i (\ket{11} + \ket{10}))
=  \frac{1}{\sqrt{2}} (\ket{0-} + i \ket{1+}) .
\end{split}
\end{equation} \qd
\end{document}